\documentclass[11pt]{article}

\usepackage[utf8]{inputenc}
\usepackage[margin=1in]{geometry}
\usepackage{amsmath}
\usepackage{amssymb}
\usepackage{amsthm}
\usepackage{mathtools}
\usepackage{bbm}
\usepackage{bm}
\usepackage{color}
\usepackage{graphicx}
\usepackage{array}
\usepackage{multirow}
\usepackage{enumitem}
\usepackage[ruled,vlined]{algorithm2e}
\usepackage{relsize}
\usepackage[dvipsnames]{xcolor}
\usepackage[linktocpage=true,
	pagebackref=true,colorlinks,
	linkcolor=BrickRed,citecolor=blue]
	{hyperref}
\usepackage{cleveref}
\usepackage[nottoc]{tocbibind}

\usepackage{tikz}
\usetikzlibrary{calc,spy}
\usepackage{subcaption}
\usepackage{xcolor}
\usepackage{float}

\usepackage{listings}
\definecolor{dkgreen}{rgb}{0,0.6,0}
\definecolor{gray}{rgb}{0.5,0.5,0.5}
\definecolor{mauve}{rgb}{0.58,0,0.82}

\usepackage{thmtools,thm-restate}
\usepackage{nicefrac}
\usepackage{makecell}

\DeclareMathOperator*{\I}{\mathbb{I}}
\DeclareMathOperator*{\E}{\mathbb{E}}

\DeclareMathOperator{\given}{\,\vert\,}
\DeclareMathOperator{\mutual}{;\,}
\DeclareMathOperator{\relative}{\,\Vert\,}
\DeclareMathOperator*{\KL}{\mathbf{D}}
\DeclareMathOperator*{\MI}{\mathbf{I}}
\DeclareMathOperator*{\TV}{\mathbf{TV}}
\let\H\relax
\DeclareMathOperator*{\H}{\mathbf{H}}

\DeclareMathOperator*{\CIC}{CIC}
\DeclareMathOperator*{\CC}{CC}

\DeclareMathOperator{\OPT}{OPT}
\DeclareMathOperator{\ALG}{ALG}

\DeclareMathOperator{\poly}{poly}
\DeclareMathOperator{\Binom}{Binom}

\newcommand{\SA}{\ensuremath{\mathsf{SA}}}
\newcommand{\XOS}{\ensuremath{\mathsf{XOS}}}
\newcommand{\SM}{\ensuremath{\mathsf{SM}}}
\newcommand{\SC}{\ensuremath{\mathsf{Succ}}}
\newcommand{\SASM}{\ensuremath{\mathsf{SA}\,\cup\,\mathsf{SM}}}
\newcommand{\SASC}{\ensuremath{\mathsf{SA}\,\cup\,\mathsf{Succ}}}
\newcommand{\XOSSM}{\ensuremath{\mathsf{XOS}\,\cup\,\mathsf{SM}}}
\newcommand{\XOSSC}{\ensuremath{\mathsf{XOS}\,\cup\,\mathsf{Succ}}}
\newcommand{\VSM}{\ensuremath{{\calV}\,\cup\,\mathsf{SM}}}
\newcommand{\VSC}{\ensuremath{{\calV}\,\cup\,\mathsf{Succ}}}
\newcommand*{\FS}[1]{\ensuremath{\textsc{FarSets}_{#1}}}
\newcommand*{\EFS}[1]{\ensuremath{\exists\textsc{FarSets}_{#1}}}

\newcommand{\vb}{\boldsymbol}

\newcommand{\calC}{\mathcal{C}}
\newcommand{\calD}{\mathcal{D}}

\newcommand{\calL}{\mathcal{L}}

\newcommand{\calP}{\mathcal{P}}
\newcommand{\calQ}{\mathcal{Q}}

\newcommand{\calS}{\mathcal{S}}
\newcommand{\calT}{\mathcal{T}}
\newcommand{\calU}{\mathcal{U}}
\newcommand{\calV}{\mathcal{V}}

\newcommand{\calX}{\mathcal{X}}
\newcommand{\calY}{\mathcal{Y}}

\let\abs\relax
\DeclarePairedDelimiter{\abs}{\lvert}{\rvert}
\DeclarePairedDelimiter{\ceil}{\lceil}{\rceil}

\let\norm\relax
\DeclarePairedDelimiter{\norm}{\lVert}{\rVert}

\newcommand{\IGNORE}[1]{}

\newcounter{note}[section]

\newcommand{\PreserveBackslash}[1]{\let\temp=\\#1\let\\=\temp}
\newcolumntype{C}[1]{>{\PreserveBackslash\centering}p{#1}}
\newcolumntype{R}[1]{>{\PreserveBackslash\raggedleft}p{#1}}
\newcolumntype{L}[1]{>{\PreserveBackslash\raggedright}p{#1}}

\newcommand{\ols}[1]{\mskip.5\thinmuskip\overline{\mskip-.5\thinmuskip {#1} \mskip-.5\thinmuskip}\mskip.5\thinmuskip} % overline short
\newcommand{\olsi}[1]{\,\overline{\!{#1}}} % overline short italic
\makeatletter
\newcommand\closure[1]{
  \tctestifnum{\count@stringtoks{#1}>1} %checks if number of chars in arg > 1 (including '\')
  {\ols{#1}} %if arg is longer than just one char, e.g. \mathbb{Q}, \mathbb{F},...
  {\olsi{#1}} %if arg is just one char, e.g. K, L,...
}
\long\def\count@stringtoks#1{\tc@earg\count@toks{\string#1}}
\long\def\count@toks#1{\the\numexpr-1\count@@toks#1.\tc@endcnt}
\long\def\count@@toks#1#2\tc@endcnt{+1\tc@ifempty{#2}{\relax}{\count@@toks#2\tc@endcnt}}
\def\tc@ifempty#1{\tc@testxifx{\expandafter\relax\detokenize{#1}\relax}}
\long\def\tc@earg#1#2{\expandafter#1\expandafter{#2}}
\long\def\tctestifnum#1{\tctestifcon{\ifnum#1\relax}}
\long\def\tctestifcon#1{#1\expandafter\tc@exfirst\else\expandafter\tc@exsecond\fi}
\long\def\tc@testxifx{\tc@earg\tctestifx}
\long\def\tctestifx#1{\tctestifcon{\ifx#1}}
\long\def\tc@exfirst#1#2{#1}
\long\def\tc@exsecond#1#2{#2}
\makeatother

\newtheorem{theorem}{Theorem}[section]

\Crefname{claim}{Claim}{Claims}
\newtheorem{proposition}[theorem]{Proposition}
\newtheorem{lemma}[theorem]{Lemma}
\newtheorem{corollary}[theorem]{Corollary}

\newtheorem{fact}[theorem]{Fact}
\Crefname{fact}{Fact}{Facts}
\theoremstyle{definition}
\newtheorem{example}[theorem]{Example}

\newtheorem{definition}[theorem]{Definition}

\newtheorem{remark}[theorem]{Remark}

\title{The Communication Complexity of Combinatorial Auctions with Additional Succinct Bidders}
\author{Frederick V. Qiu \and S.~Matthew Weinberg \and Qianfan Zhang}
\date{\today}

\begin{document}

\maketitle

\begin{abstract}
    We study the communication complexity of welfare maximization in combinatorial auctions with bidders from \emph{either} a standard valuation class (which require exponential communication to explicitly state, such as subadditive or XOS), \emph{or} arbitrary succinct valuations (which can be fully described in polynomial communication, such as single-minded). Although succinct valuations can be efficiently communicated, we show that additional succinct bidders have a nontrivial impact on communication complexity of classical combinatorial auctions. Specifically:
    
    \medskip
    
    \noindent Let $n$ be the number of subadditive/XOS bidders. We show that for \SASC{} (the union of subadditive and succinct valuations):
    \begin{itemize}
        \item There is a polynomial communication $3$-approximation algorithm.
        \item As $n \to \infty$, there is a matching $3$-hardness of approximation, which (a) is larger than the optimal approximation ratio of $2$ for \SA{}~\cite{Feige09}, and (b) holds even for \SASM{} (the union of subadditive and single-minded valuations).
        \item For all $n \geq 3$, there is a constant separation between the optimal approximation ratios for \SASM{} and \SA{} (and therefore between \SASC{} and \SA{} as well).
    \end{itemize}

    \noindent Similarly, we show that for \XOSSC{}:
    \begin{itemize}
        \item There is a polynomial communication $2$-approximation algorithm.
        \item As $n \to \infty$, there is a matching $2$-hardness of approximation, which (a) is larger than the optimal approximation ratio of $e/(e-1)$ for \XOS{}~\cite{DobzinskiNS10}, and (b) holds even for \XOSSM{}.
        \item For all $n \geq 2$, there is a constant separation between the optimal approximation ratios for \XOSSM{} and \XOS{} (and therefore between \XOSSC{} and \XOS{} as well).
    \end{itemize}
\end{abstract}

\section{Introduction}
\label{sec:intro}

In a combinatorial auction, there are a set of items $M \coloneqq [m]$ and bidders $N \coloneqq [n]$. Each bidder $i \in N$ holds a private monotone valuation function $v_i : 2^M \to \mathbb{R}_+$. We often restrict the valuations to a valuation class $\calV$ depending on the problem setting. An \emph{allocation} is an $n$-tuple of pairwise disjoint subsets of $M$, and we denote the set of all allocations by $\Sigma$. The objective is to find an allocation $\vb{A} = (A_1, \dots, A_n) \in \Sigma$ that approximately maximizes the social welfare $\sum_{i \in N} v_i(A_i)$.

Since each bidder $i$ only knows their own valuation $v_i$, two challenges arise:
\begin{enumerate}[topsep=4pt]
    \item The bidders must communicate in order to find a high welfare allocation.
    
    \item A bidder may attempt to strategically manipulate the protocol by lying about their valuation, reducing overall welfare in order to achieve higher utility for herself.
\end{enumerate}

Combinatorial auctions are well understood with respect to challenge $(1)$. For example, with fully cooperative bidders and polynomial communication, a tight $\Theta(\sqrt{m})$-approximation is known for general valuations~\cite{LehmannOS02,NisanS06,BlumrosenN05a}, a tight $2$-approximation is known for subadditive valuations~\cite{Feige09,EzraFNTW19}, a tight $e/(e-1)$-approximation is known for XOS valuations~\cite{Feige09,DobzinskiNS10}, and the optimal approximation ratio for submodular valuations is known to lie in $[2e/(2e-1), e/(e-1) - 10^{-5}]$~\cite{DobzinskiV13,FeigeV10}. However, these approximation algorithms fail if bidders may strategically manipulate the protocol.

Combinatorial auctions are also well understood with respect to challenge $(2)$. The VCG auction~\cite{Vickrey61,Clarke71,Groves73} is a mechanism by which bidders are charged prices $\vb{p}$ corresponding to the allocation $\vb{A}$, where $\vb{p}$ is carefully chosen so that bidders seeking to maximize their utility $v_i(A_i) - p_i$ are incentivized to cooperate with the protocol. Such mechanisms are called \emph{truthful}. However, the VCG auction requires exponential communication to run.

When both challenges are considered together, surprisingly little is known. The best known deterministic truthful mechanisms achieve an $\tilde{O}(m)$-approximation for general valuations and $\tilde{O}(\sqrt{m})$-approximation for subadditive/XOS/submodular valuations~\cite{QiuW24}, nearly a factor of $\sqrt{m}$ worse than their fully algorithmic (i.e., ignoring issues of incentives) counterparts. Not only that, but known lower bounds for deterministic truthful mechanisms beyond those already known for non-truthful algorithms are rare and apply only to restricted settings~\cite{AssadiKSW20, RonTWZ24}. As such, understanding the interplay between efficiency and incentive compatibility is a paradigmatic problem at the intersection of Economics and Computing.

\subsection{Combinatorial Auctions with Additional Succinct Bidders}

As of now, we have only mentioned the valuation classes submodular, XOS, and subadditive. These classes are interesting because while they provide additional structure that makes good approximation tractable, the valuation classes are still rich enough so that bidders cannot just trivially communicate their entire valuation. However, the recent work of~\cite{RonTWZ24} proved the first $3$-bidder separation between deterministic truthful mechanisms and non-truthful algorithms for the valuation class \emph{subadditive-union-single-minded}, initiating the study of what we term \emph{combinatorial auctions with additional succinct bidders}. We call a bidder succinct if their valuation can be fully expressed in polynomial communication.

Although the union of two disparate valuation classes has not received much previous attention, its study is motivated by the following.
\begin{itemize}[topsep=4pt]
    \item The approach of the $3$-bidder deterministic truthful lower bound for subadditive-union-single-minded valuations in~\cite{RonTWZ24} is the most promising avenue towards a non-trivial $\poly(m)$-bidder deterministic truthful lower bound, as it is the first to leverage the full power of the taxation complexity framework introduced in~\cite{Dobzinski16b}. The $2$-bidder lower bound for XOS valuations in~\cite{AssadiKSW20} also uses a consequence of the taxation complexity framework, but their approach cannot be generalized beyond $2$ bidders.

    \item While~\cite{RonTWZ24} is the first work to point out the potential usefulness of additional succinct bidders in proving lower bounds, the concept has shown up previously:~\cite{QiuW24} implicitly proves that for VCG-based mechanisms (a subclass of deterministic truthful mechanisms), welfare maximization over the valuation class submodular-union-single-minded is exactly as hard as welfare maximization over general valuations.\footnote{In particular, in~\cite{QiuW24}, Lemma~5.4 only requires single-minded valuations, and the remainder of the lower bound proof in Section~5 only requires submodular valuations.}~This indicates that investigating similar valuation classes could be promising for deterministic truthful lower bounds in general.

    \item Finally, the addition of succinct valuations is conceptually important precisely because of how ``deceptively uninteresting'' it is. Since they can be fully communicated, additional succinct valuations seemingly shouldn't increase the communication complexity of combinatorial auctions, but the two previous examples show there are settings where they do.
\end{itemize}

Thus, the goal of this paper is to understand combinatorial auctions with additional succinct valuations from an \emph{algorithmic} perspective, where no non-trivial prior work exists.\footnote{The only \emph{algorithmic} observation about succinct bidders necessary for~\cite{RonTWZ24} is that a constant number of single-minded bidders can be exhausted over, and therefore the communication complexity of combinatorial auctions for subadditive bidders with a constant number of additional single-minded bidders is exactly the same as with just subadditive bidders. The novelty in their work lies entirely in their lower bounds for \emph{truthful} combinatorial auctions with additional succinct bidders.} While the story so far has been heavily motivated by deterministic truthful mechanisms, it is not possible to understand the \emph{separation} between deterministic truthful mechanisms and algorithms without understanding the algorithms. We study the valuation classes subadditive-union-arbitrary-succinct and XOS-union-arbitrary-succinct (which we denote by \SASC{} and \XOSSC{}), but our lower bounds hold even for the narrower classes of subadditive-union-single-minded (\SASM{}) and XOS-union-single-minded (\XOSSM{}).

\begin{theorem}[name=,restate=ComplexitySASM] \label{thm:ComplexitySASM}
    Let there be $n$ subadditive bidders and $c$ succinct bidders. Then there exists a polynomial communication $(3-2/n)$-approximation algorithm for \SASC{}. Additionally, for all $a \leq n$, any $3(1-4a^{-1/3})$-approximation for \SASM{} uses $\min\{2^{\Omega(c/a)}, 2^{\Omega(\sqrt{m/a^3})}\}$ communication.
    
    In particular, when $c = \Omega(\sqrt{m})$, we have for any constant $\varepsilon > 0$ that a $(3-\varepsilon)$-approximation for \SASM{} requires $2^{\Omega(\sqrt{m})}$ communication.
\end{theorem}

\begin{theorem}[name=,restate=ComplexityXOSSM] \label{thm:ComplexityXOSSM}
    Let there be $n$ XOS bidders and $c$ single-minded bidders. Then there exists a polynomial communication $2$-approximation algorithm for \XOSSC{}. Additionally, for all $a \leq n$, any $2(1-4a^{-1/3})$-approximation for \XOSSM{} uses $\min\{2^{\Omega(c/a)}, 2^{\Omega(m/a^2)}\}$ communication.
    
    In particular, when $c = \Omega(m)$, we have for any constant $\varepsilon > 0$ that a $(2 - \varepsilon)$-approximation for \XOSSM{} requires $2^{\Omega(m)}$ communication.
\end{theorem}

One implication of our results pertains to the lower bounds of~\cite{RonTWZ24}. There, they show that no deterministic truthful mechanism can beat a $(1+\sqrt{3})$-approximation in polynomial communication when there are at least $3$ bidders for \SASM{}. They also show that when there are only $O(\log(m))$ bidders total, there exists a poly-communication algorithm that achieves a $2$-approximation for \SASM{},\footnote{This follows by simply exhausting over all $2^{O(\log(m))} = \poly(m)$ possible allocations to the single-minded bidders and running~\cite{Feige09}'s $2$-approximation on the remaining items for the subadditive bidders.} meaning their $3$-bidder separation extends to as many as $O(\log(m))$ bidders, because $2 < 1+\sqrt{3}$. However, because $1+\sqrt{3} < 3-o(1)$, our result shows that their separation does not extend to the general $\poly(m)$-bidder setting.

We also show that for all non-trivial\footnote{If there are two subadditive bidders, there exists a $2$-approximation algorithm for \SASM{} by \Cref{thm:ComplexitySASM}. If there is one XOS bidder, she can find the optimal allocation by herself after learning the succinct valuations.} numbers of subadditive/XOS bidders, the addition of single-minded bidders makes welfare maximization strictly more difficult.

\begin{theorem}[name=,restate=SeparationSASM] \label{thm:SeparationSASM}
    Let there be $n$ subadditive bidders and $c$ single-minded bidders. When $n \geq 3$, any $2.06$-approximation algorithm for \SASM{} uses $2^{\Omega(\min\{c, \sqrt{m}\})}$ communication.
\end{theorem}

\begin{theorem}[name=,restate=SeparationXOSSM] \label{thm:SeparationXOSSM}
    Let there be $n$ XOS bidders and $c$ single-minded bidders. When $n \geq 2$, any $(1/(1 - (1 - 1/n)^n)+0.001)$-approximation algorithm for \XOSSM{} uses $2^{\Omega(\min\{c, \sqrt{m}\})}$ communication.
\end{theorem}

\begin{theorem}[\cite{Feige09,DobzinskiNS10}] \label{thm:ApproxForSAAndXOS}
    Let $n$ be the number of bidders. There exists a polynomial communication $2$-approximation algorithm for subadditive valuations and a polynomial communication $1/(1-(1-1/n)^n)$-approximation algorithm for XOS valuations.
\end{theorem}

As noted previously, it is somewhat surprising that additional succinct bidders, who can fully share their input in polynomial communication, can complicate the \emph{communication complexity} of welfare optimization. Indeed, if there were \emph{only} succinct bidders, we could trivially maximize the welfare in polynomial communication by just asking each bidder to report their full valuation.

Finally, on the way to our hardness results, we prove a strong inapproximability result for subadditive valuations which may be of independent interest.

\begin{definition}[Scarce]
    An allocation is \emph{$t$-scarce} if all but at most $t$ bidders are allocated $\emptyset$.
\end{definition}

\begin{definition}[Strong Inapproximability] \label{def:IntroHelper}
    A valuation class $\calV$ is \emph{strongly $\alpha(\cdot)$-inapproximable in $z$ communication} if for all $t \in [n]$, beating an $(\alpha(t) \cdot n/t)$-approximation with a $t$-scarce allocation requires $z$ communication. In other words, restricting attention to \emph{any} sub-instance of $t$ bidders from the given instance retains $\alpha(t)$-inapproximability.
\end{definition}

\begin{theorem}[name=,restate=SubadditiveStronglyInapproximableThm] \label{thm:SubadditiveStronglyInapproximable}
    Subadditive valuations are strongly $(2-\I[t = 1]-O(\log(n)/\log(m)))$-inapproximable in $2^{\Omega(\sqrt{m}-n^2\log(m))}$ communication.
\end{theorem}

In other words, not only is it hard to beat a $2$-approximation for subadditive valuations, it is hard to beat a $2$-approximation for \emph{any sub-instance with at least two bidders}.

\subsection{Technical Highlights}

Below, we share some technical highlights for our upper and lower bounds.

\paragraph{\XOSSM{} Lower Bound.}

Let there be $n$ XOS bidders and $c$ single-minded bidders, and for simplicity, let $m = nc$. Our \XOSSM{} lower bound works by dividing the $m$ items into $n$ rows and $c$ columns. For some small $\varepsilon$, each XOS bidder $i$ holds a collection $\calS_i$ of subsets of row $i$ of size $\varepsilon c$, and their valuation is the rank function of the downward-closed set family defined by $\calS_i$. Each single-minded bidder $j$ just wants the items in column $j$. We weight the bidders' valuations accordingly so that the maximum welfare from the single-minded bidders is $1$, and the maximum welfare from the XOS bidders is $1$.

If there exists a set $S$ of $\varepsilon c$ columns where, slightly abusing notation, $S \in \calS_i$ for all XOS bidders $i$, then the optimal welfare is approximately $2-\varepsilon$: we can give the XOS bidders the $\varepsilon c$ columns $S$, and the single-minded bidders everything else.

\newcommand{\Nrows}{4}
\newcommand{\Ncols}{6}

\newcommand{\BarWidth}{0.8}   % bar width & height in grid-units
\newcommand{\Gap}{0.2}   % gap between columns in grid-units
\newcommand{\Unit}{0.6}  % 1 grid-unit
\pgfmathsetmacro{\Grid}{\BarWidth+\Gap}
\colorlet{SMColor}{blue!60!black}
\colorlet{XOSColor}{green!60!black}
\colorlet{AColor}{green!100!lime!80!black}
\colorlet{BColor}{red!100!magenta!100!black}
\colorlet{CColor}{orange!100!black}
\colorlet{DColor}{cyan!100!black}
\colorlet{Gray}{gray}

\def\XOSColors{{"AColor","BColor","CColor","DColor"}}

\tikzset{
  every picture/.style = {x=\Unit cm, y=\Unit cm},
  prefCol/.style = {draw=SMColor,  draw=none,
                    fill=blue!15,        fill opacity=.5,
                    rounded corners=6pt, line width=1pt},
  prefRow/.style = {draw=XOSColor, dashed,
                    fill=green!15,       fill opacity=.5,
                    rounded corners=6pt, line width=1pt},
  colBar/.style  = {fill=SMColor,  rounded corners=6pt,
                    draw=none, opacity=.65},
  rowBar/.style  = {fill=XOSColor, rounded corners=6pt,
                    draw=none, opacity=.75},
  item/.style    = {circle, draw, fill=white,
                    inner sep=2.1pt, thick}
}

% draw item-grid
\newcommand{\itemgrid}{%
  \foreach \i in {0,...,\numexpr\Nrows-1\relax}{%
    \foreach \j in {0,...,\numexpr\Ncols-1\relax}{%
      \pgfmathsetmacro{\Xc}{\j*\Grid + \Gap/2 + \BarWidth/2}%
      \pgfmathsetmacro{\Yc}{-\i}%
      \node[item] at (\Xc,\Yc) {};}}}%

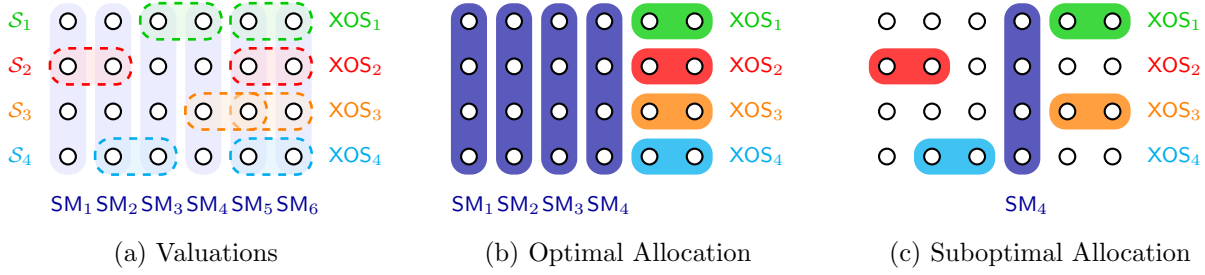
\begin{figure}[H]
\centering

% ---------- (a) Valuations ----------
\begin{subfigure}[t]{.32\linewidth}
\centering
\begin{tikzpicture}
  \pgfmathsetmacro{\Xleft}{\Gap/2}
  \pgfmathsetmacro{\Xright}{\Ncols*\Grid - \Gap/2}
  \pgfmathsetmacro{\Ytop}{ \BarWidth/2}
  \pgfmathsetmacro{\Ybot}{-(\Nrows-1) - \BarWidth/2}
  \pgfmathsetmacro{\MidX}{0.5*(\Xleft+\Xright)}
  \pgfmathsetmacro{\MidY}{-0.5*(\Nrows-1)}

  % columns
  \foreach \j in {0,...,\numexpr\Ncols-1\relax}{
    \pgfmathsetmacro{\XL}{\j*\Grid + \Gap/2}
    \draw[prefCol] (\XL,\Ytop) rectangle ++(\BarWidth,\Ybot-\Ytop);
  }

  % rows
  \pgfmathsetmacro{\PrefL}{4*\Grid + \Gap/2}
  \pgfmathsetmacro{\PrefW}{2*\BarWidth + \Gap}
  \foreach \i [evaluate=\i as \usecolor using {\XOSColors[\i]}] in {0,...,\numexpr\Nrows-1\relax}{
    \pgfmathsetmacro{\Yc}{-\i}
    \draw[prefRow, draw=\usecolor, fill=\usecolor!15] (\PrefL,\Yc-\BarWidth/2) rectangle ++(\PrefW,\BarWidth);
  }
  \pgfmathsetmacro{\RowW}{\PrefW}
  \draw[prefRow, draw=AColor, fill=AColor!15] (2*\Grid+\Gap/2,-0-\BarWidth/2) rectangle ++(\RowW,\BarWidth);
  \draw[prefRow, draw=BColor, fill=BColor!15] (0*\Grid+\Gap/2,-1-\BarWidth/2) rectangle ++(\RowW,\BarWidth);
  \draw[prefRow, draw=CColor, fill=CColor!15] (3*\Grid+\Gap/2,-2-\BarWidth/2) rectangle ++(\RowW,\BarWidth);
  \draw[prefRow, draw=DColor, fill=DColor!15] (1*\Grid+\Gap/2,-3-\BarWidth/2) rectangle ++(\RowW,\BarWidth);

  \itemgrid

  % labels
  \foreach \i [evaluate=\i as \usecolor using {\XOSColors[\i-1]}] in {1,...,\Nrows}{
    \pgfmathsetmacro{\rowY}{1-\i}
    \node[anchor=west, scale=1, text=\usecolor]
          at (\Xright+0.15, \rowY) {$\scriptstyle\XOS_\i$};
  }
  \foreach \i [evaluate=\i as \usecolor using {\XOSColors[\i-1]}] in {1,...,\Nrows}{
    \pgfmathsetmacro{\rowY}{1-\i}
    \node[anchor=east, scale=1, text=\usecolor]
          at (\Xleft-0.15, \rowY) {$\scriptstyle\calS_\i$};
  }
  \foreach \j in {1,...,\Ncols}{
    \pgfmathsetmacro{\colX}{\j-0.5+\Gap/2}
    \node[anchor=north, scale=1, text=SMColor]
          at (\colX, \Ybot-0.25) {$\scriptstyle\SM_{\j}$};
  }
\end{tikzpicture}
\caption{Valuations}
\end{subfigure}
\hfill
% ---------- (b) Optimal allocation ----------
\begin{subfigure}[t]{.32\linewidth}
\centering
\begin{tikzpicture}
  \pgfmathsetmacro{\Xleft}{\Gap/2}
  \pgfmathsetmacro{\Xright}{\Ncols*\Grid - \Gap/2}
  \pgfmathsetmacro{\Ytop}{ \BarWidth/2}
  \pgfmathsetmacro{\Ybot}{-(\Nrows-1) - \BarWidth/2}
  \pgfmathsetmacro{\MidX}{0.5*(\Xleft+\Xright)}
  \pgfmathsetmacro{\MidY}{-0.5*(\Nrows-1)}

  % columns
  \foreach \j in {0,...,3}{
    \pgfmathsetmacro{\XL}{\j*\Grid + \Gap/2}
    \draw[colBar] (\XL,\Ytop) rectangle ++(\BarWidth,\Ybot-\Ytop);
  }

  % rows
  \pgfmathsetmacro{\AllocL}{4*\Grid + \Gap/2}
  \pgfmathsetmacro{\RowW}{2*\BarWidth + \Gap}
  \foreach \i [evaluate=\i as \usecolor using {\XOSColors[\i]}] in {0,...,3}{
    \pgfmathsetmacro{\Yc}{-\i}
    \draw[rowBar, fill=\usecolor] (\AllocL,\Yc-\BarWidth/2) rectangle ++(\RowW,\BarWidth);
  }

  \itemgrid

  % labels
  \foreach \i [evaluate=\i as \usecolor using {\XOSColors[\i-1]}] in {1,...,\Nrows}{
    \pgfmathsetmacro{\rowY}{1-\i}
    \node[anchor=west, scale=1, text=\usecolor]
          at (\Xright+0.15, \rowY) {$\scriptstyle\XOS_\i$};
  }
  \foreach \j in {1,...,\numexpr\Ncols-2\relax}{
    \pgfmathsetmacro{\colX}{\j-0.5+\Gap/2}
    \node[anchor=north, scale=1, text=SMColor]
          at (\colX, \Ybot-0.25) {$\scriptstyle\SM_{\j}$};
  }
  % \foreach \j in {5,6}{
  %   \pgfmathsetmacro{\colX}{\j-0.5+\Gap/2}
  %   \node[anchor=north, scale=1, text=Gray]
  %         at (\colX, \Ybot-0.25) {$\scriptstyle\SM_{\j}$};
  % }
\end{tikzpicture}
\caption{Optimal Allocation}
\end{subfigure}
\hfill
% ---------- (c) Sub-optimal allocation ----------
\begin{subfigure}[t]{.32\linewidth}
\centering
\begin{tikzpicture}
  \pgfmathsetmacro{\Xleft}{\Gap/2}
  \pgfmathsetmacro{\Xright}{\Ncols*\Grid - \Gap/2}
  \pgfmathsetmacro{\Ytop}{ \BarWidth/2}
  \pgfmathsetmacro{\Ybot}{-(\Nrows-1) - \BarWidth/2}
  \pgfmathsetmacro{\MidX}{0.5*(\Xleft+\Xright)}
  \pgfmathsetmacro{\MidY}{-0.5*(\Nrows-1)}

  \pgfmathsetmacro{\RowW}{2*\BarWidth + \Gap}

  % columns
  \pgfmathsetmacro{\XLsm}{3*\Grid + \Gap/2}
  \draw[colBar] (\XLsm,\Ytop) rectangle ++(\BarWidth,\Ybot-\Ytop);

  % rows
  \draw[rowBar, fill=AColor] (4*\Grid+\Gap/2,-0-\BarWidth/2) rectangle ++(\RowW,\BarWidth);
  \draw[rowBar, fill=BColor] (0*\Grid+\Gap/2,-1-\BarWidth/2) rectangle ++(\RowW,\BarWidth);
  \draw[rowBar, fill=CColor] (4*\Grid+\Gap/2,-2-\BarWidth/2) rectangle ++(\RowW,\BarWidth);
  \draw[rowBar, fill=DColor] (1*\Grid+\Gap/2,-3-\BarWidth/2) rectangle ++(\RowW,\BarWidth);

  \itemgrid

  % labels
  \foreach \i [evaluate=\i as \usecolor using {\XOSColors[\i-1]}] in {1,...,\Nrows}{
    \pgfmathsetmacro{\rowY}{1-\i}
    \node[anchor=west, scale=1, text=\usecolor]
          at (\Xright+0.15, \rowY) {$\scriptstyle\XOS_\i$};
  }
  % \foreach \j in {1,...,\Ncols}{
  %   \pgfmathsetmacro{\colX}{\j-0.5+\Gap/2}
  %   \node[anchor=north, scale=1, text=Gray]
  %         at (\colX, \Ybot-0.25) {$\scriptstyle\SM_{\j}$};
  % }
  \foreach \j in {4}{
    \pgfmathsetmacro{\colX}{\j-0.5+\Gap/2}
    \node[anchor=north, scale=1, text=SMColor]
          at (\colX, \Ybot-0.25) {$\scriptstyle\SM_{\j}$};
  }
\end{tikzpicture}
\caption{Suboptimal Allocation}
\end{subfigure}

\caption{An illustrative example of a hard instance for the \XOSSM{} lower bound. Here we have $n=4$ XOS bidders, $c=6$ single-minded bidders, and $m=nc=24$ items. The valuation for $\XOS_i$ is defined by the collection of items $\calS_i$ over row $i$, while $\SM_j$ is single-minded for the items in column $j$.
The optimal allocation can only be achieved by allocating the common columns $\{5, 6\}$ to the XOS bidders and allocating the rest to the single-minded bidders. On the other hand, if the algorithm cannot identify the common column set $\{5, 6\}$, it must either fail to satisfy many single-minded bidders as in (c), or fail to satisfy many XOS bidders.}
\label{fig:xos-sm-hard-instance}
\end{figure}

However, so long as $\abs{\calS_i}$ is exponentially large, it is impossible for an algorithm using polynomial communication to find a small column set which gives the XOS bidders high welfare. If the algorithm gives the XOS bidders a large column set, then the single-minded bidders get low welfare.

Interestingly, this construction does not feature any direct contention among the XOS bidders. Instead, the high level takeaway is that single-minded bidders can be used to correlate XOS bidders' allocations in a way that other XOS valuations cannot, and simply knowing the nature of this correlation does not help a polynomial communication algorithm.

\paragraph{\SASM{} Lower Bound.}

Our \SASM{} construction leverages the single-minded bidders like the \XOSSM{} lower bound, but additionally uses the $2$-inapproximability of subadditive valuations.

More generally, let $\calV$ be a valuation class closed under addition and point-wise maximum (these are the minimal properties needed to create an \XOSSM{}-style hard instance; indeed, observe that \XOS{} is the smallest non-trivial valuation class closed under addition and point-wise maximum). Further, suppose that $\calV$ cannot be $\alpha$-approximated in polynomial communication.

Then repeat the construction of the \XOSSM{} hard instance, but instead of each column having $n$ items, make each column have $nc$ items (and let $m = nc^2$), and instead of each $\calV$-type bidder wanting a single item from each column, the $\calV$-type bidders should form a hard instance within each column. Then each $\calV$-type bidder's overall value adds up across columns in an XOS-like fashion.

Now any efficient algorithm faces two challenges:
\begin{enumerate}[topsep=4pt]
    \item As in the \XOSSM{} construction, the algorithm cannot hope to find a small column set satisfying the $\calV$-type bidders.

    \item Additionally, the algorithm cannot achieve an $\alpha$-approximation of the $\calV$-type bidders within each column.
\end{enumerate}

Therefore, by letting the $\calV$-type bidders contribute welfare $\alpha$ and the single-minded bidders contribute welfare $1$, the optimal welfare is approximately $1+\alpha-\varepsilon$, but an efficient algorithm cannot beat welfare $1$.

There is one problem: each $\calV$-type bidder only gets value from at most $\varepsilon c$ columns. But this means that we can split up the $\calV$-type bidders among the columns so that in each column, we only need to consider allocating items to $\varepsilon n$ of the $\calV$-type bidders, instead of all $n$ of them. If this easier sub-problem is no longer $\alpha$-inapproximable, then an algorithm could hope to beat welfare $1$, and therefore beat a $(1+\alpha-\varepsilon)$-approximation.

\renewcommand{\Nrows}{4}
\renewcommand{\Ncols}{6}

\renewcommand{\BarWidth}{0.8}   % bar width & height in grid-units
\renewcommand{\Gap}{0.2}   % gap between columns in grid-units
\renewcommand{\Unit}{0.6}  % 1 grid-unit
\pgfmathsetmacro{\Grid}{\BarWidth+\Gap}
\pgfmathsetmacro{\RowW}{2*\BarWidth+\Gap}

% Dot vertical spacing and horizontal dot spacing inside logical cell
\def\dotVspace{0.3333}
\def\dotHoffset{0.4} % approx width of dot column (2 columns per logical column -> BarWidth/2 = 0.4)

\colorlet{SMColor}{blue!60!black}
\colorlet{AColor}{green!100!lime!80!black}
\colorlet{BColor}{red!100!magenta!100!black}
\colorlet{CColor}{orange!100!black}
\colorlet{DColor}{cyan!100!black}
\colorlet{Gray}{gray}
\colorlet{White}{white}

\def\colorA{AColor}
\def\colorB{BColor}
\def\colorC{CColor}
\def\colorD{DColor}
\def\colorN{White}

\tikzset{
  every picture/.style = {x=\Unit cm, y=\Unit cm},
  prefCol/.style = {draw=SMColor,  draw=none,
                    fill=blue!15,        fill opacity=.5,
                    rounded corners=2pt, line width=1pt},
  colBar/.style  = {fill=SMColor,  rounded corners=2pt,
                    draw=none, opacity=.65},
  item/.style    = {circle, draw, fill=white,
                    inner sep=2.1pt, thick}
}

% Confidence interval bar macro
\newcommand{\confidencebaratcols}[3]{%
  % #1 = left column index (j)
  % #2 = row index (i)
  % #3 = color (e.g. XOSColor)

  \pgfmathsetmacro{\xbarcenter}{(#1 + 1)*\Grid}
  \pgfmathsetmacro{\ycenter}{-#2}
  \pgfmathsetmacro{\halfwidth}{0.5*\RowW}
  \pgfmathsetmacro{\tickheight}{0.4}  % full row height

  \draw[#3, very thick] ({\xbarcenter - \halfwidth}, {\ycenter}) -- ({\xbarcenter + \halfwidth}, {\ycenter});
  \draw[#3, very thick] ({\xbarcenter - \halfwidth}, {\ycenter - \tickheight}) -- ({\xbarcenter - \halfwidth}, {\ycenter + \tickheight});
  \draw[#3, very thick] ({\xbarcenter + \halfwidth}, {\ycenter - \tickheight}) -- ({\xbarcenter + \halfwidth}, {\ycenter + \tickheight});
}

% draw 2x3 item-grid of white dots
\renewcommand{\itemgrid}{%
  \foreach \i in {0,...,\numexpr\Nrows-1\relax}{%
    \foreach \j in {0,...,\numexpr\Ncols-1\relax}{%
      \pgfmathsetmacro{\Xc}{\j*\Grid + \Gap/2 + \BarWidth/2}%
      \pgfmathsetmacro{\Yc}{-\i}%
      \foreach \dx/\dy in {-0.2/\dotVspace, 0.2/\dotVspace, -0.2/0, 0.2/0, -0.2/-\dotVspace, 0.2/-\dotVspace}{%
        \node[item, scale=0.6] at ({\Xc+\dx}, {\Yc+\dy}) {};
      }%
    }%
  }%
}

\begin{figure}[H]
\centering

% ---------- (a) Valuations ----------
\begin{subfigure}[t]{.32\linewidth}
\centering
\begin{tikzpicture}
  \pgfmathsetmacro{\Ytop}{ \BarWidth/2}
  \pgfmathsetmacro{\Ybot}{-(\Nrows-1) - \BarWidth/2}
  \pgfmathsetmacro{\Xright}{\Ncols*\Grid - \Gap/2}
  \pgfmathsetmacro{\Xleft}{\Gap/2}

  % Columns (blue preference)
  \foreach \j in {0,...,\numexpr\Ncols-1\relax}{
    \pgfmathsetmacro{\XL}{\j*\Grid + \Gap/2}
    \draw[fill=blue!15, opacity=0.5, rounded corners=2pt, draw=none] (\XL,\Ytop + 0.1) rectangle ++(\BarWidth,\Ybot-\Ytop - 0.2);
  }

  \itemgrid

  % Rows
  \confidencebaratcols{2}{0}{AColor}
  \confidencebaratcols{4}{0}{AColor}
  \confidencebaratcols{0}{1}{BColor}
  \confidencebaratcols{4}{1}{BColor}
  \confidencebaratcols{3}{2}{CColor}
  \confidencebaratcols{4}{2}{CColor}
  \confidencebaratcols{1}{3}{DColor}
  \confidencebaratcols{4}{3}{DColor}

  % Labels
  \node[anchor=west, scale=1, text=AColor]
        at (\Xright+0.15, 0) {$\scriptstyle\SA_{1}$};
  \node[anchor=west, scale=1, text=BColor]
        at (\Xright+0.15, -1) {$\scriptstyle\SA_{2}$};
  \node[anchor=west, scale=1, text=CColor]
        at (\Xright+0.15, -2) {$\scriptstyle\SA_{3}$};
  \node[anchor=west, scale=1, text=DColor]
        at (\Xright+0.15, -3) {$\scriptstyle\SA_{4}$};
  \node[anchor=east, scale=1, text=AColor]
        at (\Xleft-0.15, 0) {$\scriptstyle\calS_{1}$};
  \node[anchor=east, scale=1, text=BColor]
        at (\Xleft-0.15, -1) {$\scriptstyle\calS_{2}$};
  \node[anchor=east, scale=1, text=CColor]
        at (\Xleft-0.15, -2) {$\scriptstyle\calS_{3}$};
  \node[anchor=east, scale=1, text=DColor]
        at (\Xleft-0.15, -3) {$\scriptstyle\calS_{4}$};
  \foreach \j in {1,...,\Ncols}{
    \pgfmathsetmacro{\colX}{\j-0.5+\Gap/2}
    \node[anchor=north, scale=1, text=SMColor]
          at (\colX, \Ybot-0.25) {$\scriptstyle\SM_{\j}$};
  }
\end{tikzpicture}
\caption{Valuations}
\end{subfigure}
\hfill
% ---------- (b) Optimal allocation ----------
\begin{subfigure}[t]{.32\linewidth}
\centering

\begin{tikzpicture}
  \pgfmathsetmacro{\Ytop}{ \BarWidth/2}
  \pgfmathsetmacro{\Ybot}{-(\Nrows-1) - \BarWidth/2}
  \pgfmathsetmacro{\Xright}{\Ncols*\Grid - \Gap/2}
  \pgfmathsetmacro{\Xleft}{\Gap/2}

  % Column bars (blue)
  \foreach \j in {0,...,3}{
    \pgfmathsetmacro{\XL}{\j*\Grid + \Gap/2}
    \draw[colBar] (\XL,\Ytop + 0.1) rectangle ++(\BarWidth,\Ybot-\Ytop - 0.2);
  }

  % --- New: Per-dot shading in logical columns 4 and 5 ---

  % The 12x2 pattern of colors, stored as macro:
  % Each entry: A, B, C, or D
  % We'll define a macro to get color by row,col:
  \newcommand{\getcolor}[2]{%
    \ifnum#1=0 \ifnum#2=0 \colorA \else \colorA \fi\fi
    \ifnum#1=1 \ifnum#2=0 \colorA \else \colorA \fi\fi
    \ifnum#1=2 \ifnum#2=0 \colorA \else \colorB \fi\fi
    \ifnum#1=3 \ifnum#2=0 \colorA \else \colorB \fi\fi
    \ifnum#1=4 \ifnum#2=0 \colorB \else \colorB \fi\fi
    \ifnum#1=5 \ifnum#2=0 \colorB \else \colorC \fi\fi
    \ifnum#1=6 \ifnum#2=0 \colorB \else \colorC \fi\fi
    \ifnum#1=7 \ifnum#2=0 \colorC \else \colorC \fi\fi
    \ifnum#1=8 \ifnum#2=0 \colorD \else \colorC \fi\fi
    \ifnum#1=9 \ifnum#2=0 \colorD \else \colorC \fi\fi
    \ifnum#1=10 \ifnum#2=0 \colorD \else \colorD \fi\fi
    \ifnum#1=11 \ifnum#2=0 \colorD \else \colorD \fi\fi
  }

  % For each logical column 4 and 5
  \foreach \col in {4,5}{

    % Left X of logical column
    \pgfmathsetmacro{\XLcol}{\col*\Grid + \Gap/2}

    % Loop over dot rows (0 to 11)
    \foreach \row in {0,...,11}{
      % Loop over dot columns inside logical column (0 or 1)
      \foreach \dotcol in {0,1}{

        % Compute Y bottom of rectangle (since rows go downward)
        \pgfmathsetmacro{\Yrect}{-\row*\dotVspace + \dotVspace/2}

        % Compute X left of dot column
        \pgfmathsetmacro{\Xrect}{\XLcol + \dotcol*\dotHoffset}

        % Get color for this cell
        \edef\cellcolor{\getcolor{\row}{\dotcol}}

        % Draw rectangle for cell
        \draw[fill=\cellcolor, fill opacity=0.75, draw=none, rounded corners=0pt]
          (\Xrect, \Yrect) rectangle ++(\dotHoffset, \dotVspace);
      }
    }
  }

  % Draw the dots grid as before
  \itemgrid

  % Labels

  % Labels
  \node[anchor=west, scale=1, text=AColor]
        at (\Xright+0.15, 0) {$\scriptstyle\SA_{1}$};
  \node[anchor=west, scale=1, text=BColor]
        at (\Xright+0.15, -1) {$\scriptstyle\SA_{2}$};
  \node[anchor=west, scale=1, text=CColor]
        at (\Xright+0.15, -2) {$\scriptstyle\SA_{3}$};
  \node[anchor=west, scale=1, text=DColor]
        at (\Xright+0.15, -3) {$\scriptstyle\SA_{4}$};
  \foreach \j in {1,...,\numexpr\Ncols-2\relax}{
    \pgfmathsetmacro{\colX}{\j-0.5+\Gap/2}
    \node[anchor=north, scale=1, text=SMColor]
          at (\colX, \Ybot-0.25) {$\scriptstyle\SM_{\j}$};
  }
\end{tikzpicture}
\caption{Optimal Allocation}
\end{subfigure}
\hfill
% ---------- (c) Sub-optimal allocation ----------
\begin{subfigure}[t]{.32\linewidth}
\centering
\begin{tikzpicture}
  \pgfmathsetmacro{\Ytop}{ \BarWidth/2}
  \pgfmathsetmacro{\Ybot}{-(\Nrows-1) - \BarWidth/2}
  \pgfmathsetmacro{\Xright}{\Ncols*\Grid - \Gap/2}
  \pgfmathsetmacro{\Xleft}{\Gap/2}

  % Column bar
  \pgfmathsetmacro{\XLsm}{3*\Grid + \Gap/2}
  \draw[colBar] (\XLsm,\Ytop + 0.1) rectangle ++(\BarWidth,\Ybot-\Ytop - 0.2);

  % --- New: Per-dot shading in logical columns 4 and 5 ---

  % The 12x2 pattern of colors, stored as macro:
  % Each entry: A, B, C, or D
  % We'll define a macro to get color by row,col:
  \newcommand{\getcolorOne}[2]{%
    \ifnum#1=0 \ifnum#2=0 \colorN \else \colorN \fi\fi
    \ifnum#1=1 \ifnum#2=0 \colorN \else \colorB \fi\fi
    \ifnum#1=2 \ifnum#2=0 \colorN \else \colorB \fi\fi
    \ifnum#1=3 \ifnum#2=0 \colorN \else \colorB \fi\fi
    \ifnum#1=4 \ifnum#2=0 \colorB \else \colorB \fi\fi
    \ifnum#1=5 \ifnum#2=0 \colorB \else \colorB \fi\fi
    \ifnum#1=6 \ifnum#2=0 \colorB \else \colorB \fi\fi
    \ifnum#1=7 \ifnum#2=0 \colorN \else \colorB \fi\fi
    \ifnum#1=8 \ifnum#2=0 \colorN \else \colorN \fi\fi
    \ifnum#1=9 \ifnum#2=0 \colorN \else \colorN \fi\fi
    \ifnum#1=10 \ifnum#2=0 \colorN \else \colorN \fi\fi
    \ifnum#1=11 \ifnum#2=0 \colorN \else \colorN \fi\fi
  }
  \newcommand{\getcolorTwo}[2]{%
    \ifnum#1=0 \ifnum#2=0 \colorB \else \colorB \fi\fi
    \ifnum#1=1 \ifnum#2=0 \colorB \else \colorB \fi\fi
    \ifnum#1=2 \ifnum#2=0 \colorB \else \colorB \fi\fi
    \ifnum#1=3 \ifnum#2=0 \colorD \else \colorB \fi\fi
    \ifnum#1=4 \ifnum#2=0 \colorD \else \colorB \fi\fi
    \ifnum#1=5 \ifnum#2=0 \colorD \else \colorB \fi\fi
    \ifnum#1=6 \ifnum#2=0 \colorD \else \colorB \fi\fi
    \ifnum#1=7 \ifnum#2=0 \colorD \else \colorD \fi\fi
    \ifnum#1=8 \ifnum#2=0 \colorD \else \colorD \fi\fi
    \ifnum#1=9 \ifnum#2=0 \colorD \else \colorD \fi\fi
    \ifnum#1=10 \ifnum#2=0 \colorD \else \colorB \fi\fi
    \ifnum#1=11 \ifnum#2=0 \colorD \else \colorB \fi\fi
  }
  \newcommand{\getcolorThree}[2]{%
    \ifnum#1=0 \ifnum#2=0 \colorN \else \colorN \fi\fi
    \ifnum#1=1 \ifnum#2=0 \colorN \else \colorN \fi\fi
    \ifnum#1=2 \ifnum#2=0 \colorN \else \colorN \fi\fi
    \ifnum#1=3 \ifnum#2=0 \colorN \else \colorN \fi\fi
    \ifnum#1=4 \ifnum#2=0 \colorN \else \colorN \fi\fi
    \ifnum#1=5 \ifnum#2=0 \colorN \else \colorN \fi\fi
    \ifnum#1=6 \ifnum#2=0 \colorD \else \colorN \fi\fi
    \ifnum#1=7 \ifnum#2=0 \colorD \else \colorN \fi\fi
    \ifnum#1=8 \ifnum#2=0 \colorD \else \colorD \fi\fi
    \ifnum#1=9 \ifnum#2=0 \colorD \else \colorD \fi\fi
    \ifnum#1=10 \ifnum#2=0 \colorD \else \colorD \fi\fi
    \ifnum#1=11 \ifnum#2=0 \colorD \else \colorD \fi\fi
  }
  \newcommand{\getcolorFiveSix}[2]{%
    \ifnum#1=0 \ifnum#2=0 \colorA \else \colorA \fi\fi
    \ifnum#1=1 \ifnum#2=0 \colorC \else \colorA \fi\fi
    \ifnum#1=2 \ifnum#2=0 \colorC \else \colorA \fi\fi
    \ifnum#1=3 \ifnum#2=0 \colorC \else \colorA \fi\fi
    \ifnum#1=4 \ifnum#2=0 \colorC \else \colorA \fi\fi
    \ifnum#1=5 \ifnum#2=0 \colorC \else \colorA \fi\fi
    \ifnum#1=6 \ifnum#2=0 \colorC \else \colorA \fi\fi
    \ifnum#1=7 \ifnum#2=0 \colorA \else \colorA \fi\fi
    \ifnum#1=8 \ifnum#2=0 \colorA \else \colorC \fi\fi
    \ifnum#1=9 \ifnum#2=0 \colorA \else \colorC \fi\fi
    \ifnum#1=10 \ifnum#2=0 \colorC \else \colorC \fi\fi
    \ifnum#1=11 \ifnum#2=0 \colorC \else \colorC \fi\fi
  }

  % Left X of logical column
  \pgfmathsetmacro{\XLcol}{0*\Grid + \Gap/2}

  % Loop over dot rows (0 to 11)
  \foreach \row in {0,...,11}{
    % Loop over dot columns inside logical column (0 or 1)
    \foreach \dotcol in {0,1}{

      % Compute Y bottom of rectangle (since rows go downward)
      \pgfmathsetmacro{\Yrect}{-\row*\dotVspace + \dotVspace/2}

      % Compute X left of dot column
      \pgfmathsetmacro{\Xrect}{\XLcol + \dotcol*\dotHoffset}

      % Get color for this cell
      \edef\cellcolor{\getcolorOne{\row}{\dotcol}}

      % Draw rectangle for cell
      \draw[fill=\cellcolor, fill opacity=0.75, draw=none, rounded corners=0pt]
        (\Xrect, \Yrect) rectangle ++(\dotHoffset, \dotVspace);
    }
  }

  % Left X of logical column
  \pgfmathsetmacro{\XLcol}{1*\Grid + \Gap/2}

  % Loop over dot rows (0 to 11)
  \foreach \row in {0,...,11}{
    % Loop over dot columns inside logical column (0 or 1)
    \foreach \dotcol in {0,1}{

      % Compute Y bottom of rectangle (since rows go downward)
      \pgfmathsetmacro{\Yrect}{-\row*\dotVspace + \dotVspace/2}

      % Compute X left of dot column
      \pgfmathsetmacro{\Xrect}{\XLcol + \dotcol*\dotHoffset}

      % Get color for this cell
      \edef\cellcolor{\getcolorTwo{\row}{\dotcol}}

      % Draw rectangle for cell
      \draw[fill=\cellcolor, fill opacity=0.75, draw=none, rounded corners=0pt]
        (\Xrect, \Yrect) rectangle ++(\dotHoffset, \dotVspace);
    }
  }

  % Left X of logical column
  \pgfmathsetmacro{\XLcol}{2*\Grid + \Gap/2}

  % Loop over dot rows (0 to 11)
  \foreach \row in {0,...,11}{
    % Loop over dot columns inside logical column (0 or 1)
    \foreach \dotcol in {0,1}{

      % Compute Y bottom of rectangle (since rows go downward)
      \pgfmathsetmacro{\Yrect}{-\row*\dotVspace + \dotVspace/2}

      % Compute X left of dot column
      \pgfmathsetmacro{\Xrect}{\XLcol + \dotcol*\dotHoffset}

      % Get color for this cell
      \edef\cellcolor{\getcolorThree{\row}{\dotcol}}

      % Draw rectangle for cell
      \draw[fill=\cellcolor, fill opacity=0.75, draw=none, rounded corners=0pt]
        (\Xrect, \Yrect) rectangle ++(\dotHoffset, \dotVspace);
    }
  }

  % For each logical column 4 and 5
  \foreach \col in {4,5}{

    % Left X of logical column
    \pgfmathsetmacro{\XLcol}{\col*\Grid + \Gap/2}

    % Loop over dot rows (0 to 11)
    \foreach \row in {0,...,11}{
      % Loop over dot columns inside logical column (0 or 1)
      \foreach \dotcol in {0,1}{

        % Compute Y bottom of rectangle (since rows go downward)
        \pgfmathsetmacro{\Yrect}{-\row*\dotVspace + \dotVspace/2}

        % Compute X left of dot column
        \pgfmathsetmacro{\Xrect}{\XLcol + \dotcol*\dotHoffset}

        % Get color for this cell
        \edef\cellcolor{\getcolorFiveSix{\row}{\dotcol}}

        % Draw rectangle for cell
        \draw[fill=\cellcolor, fill opacity=0.75, draw=none, rounded corners=0pt]
          (\Xrect, \Yrect) rectangle ++(\dotHoffset, \dotVspace);
      }
    }
  }

  \itemgrid

  % Labels
  \node[anchor=west, scale=1, text=AColor]
        at (\Xright+0.15, 0) {$\scriptstyle\SA_{1}$};
  \node[anchor=west, scale=1, text=BColor]
        at (\Xright+0.15, -1) {$\scriptstyle\SA_{2}$};
  \node[anchor=west, scale=1, text=CColor]
        at (\Xright+0.15, -2) {$\scriptstyle\SA_{3}$};
  \node[anchor=west, scale=1, text=DColor]
        at (\Xright+0.15, -3) {$\scriptstyle\SA_{4}$};
  \foreach \j in {4}{
    \pgfmathsetmacro{\colX}{\j-0.5+\Gap/2}
    \node[anchor=north, scale=1, text=SMColor]
          at (\colX, \Ybot-0.25) {$\scriptstyle\SM_{\j}$};
  }
\end{tikzpicture}
\caption{Suboptimal Allocation}
\end{subfigure}

\caption{The \SASM{} hard instance. Instead of wanting a single item from each column like in the \XOSSM{} hard instance, each subadditive bidder has a complex valuation over the items in each column so that the optimal allocation among subadditive bidders is hard to find. \\
\hspace*{18pt} It is possible for a suboptimal allocation like (c) to split the subadditive bidders across columns so that each column is ``responsible'' for giving value to significantly fewer subadditive bidders, since each subadditive bidder only wants an $\varepsilon$ fraction of the columns. Therefore, the standard notion of $2$-inapproximability that subadditive valuations are known to satisfy does not suffice: we need a stronger notion where almost every sub-instance of the bidders is $2$-inapproximable as well. In this example, there is only one bidder in columns $1$ and $3$, so finding an optimal allocation for it is trivial. If we could also efficiently find an optimal allocation between any two subadditive bidders, then we could recover the full subadditive welfare, meaning the instance is no harder to approximate than the simpler \XOSSM{} hard instance. \\
\hspace*{18pt} It turns out that subadditive valuations satisfy a notion of \emph{strong $2$-inapproximability} that allows for the \SASM{} hard instance to work as intended. Note that such a property is not a given; for example, XOS valuations do \emph{not} satisfy any such strong inapproximability notion, despite being $e/(e-1)$-inapproximable in the standard sense. This is why the \XOSSM{} hard instance cannot be made meaningfully harder with more complexity within each column.}
\label{fig:sa-sm-hard-instance}
\end{figure}
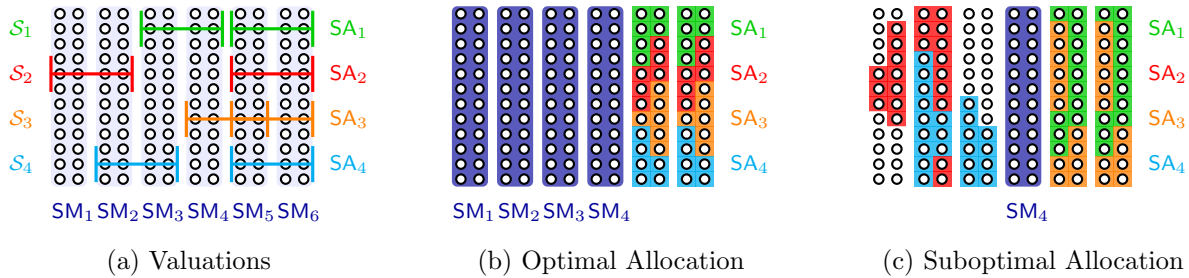

Our key technical contribution here is a notion of \emph{strong inapproximability}, where a hard instance remains hard even when all but a not-too-small subset of the bidders are removed. This notion and the following result for subadditive valuations may be of independent interest for other communication lower bounds.

We show that subadditive valuations are strongly $(2-\varepsilon)$-inapproximable. In particular, for all sub-polynomial $n$ (a slightly weaker result holds when $n = \poly(m)$), there exist $n$ subadditive valuations whose welfare is $2n$, but no polynomial communication algorithm can find an allocation where $(1)$ no more than $t \in [2, n]$ bidders receive a non-empty set of items and $(2)$ the welfare is at least $(1+\varepsilon)t$. The construction is a natural generalization of the $2$-bidder construction in~\cite{EzraFNTW19} to $n$ bidders, but the proof approach is qualitatively different;~\cite{EzraFNTW19} uses \emph{internal information complexity} for their lower bound, but we must use \emph{external} information complexity because the internal information complexity can be $0$ for any $> 2$-player communication protocol.

The \VSM{} construction above with $\calV = \SA{}$ then yields a $(3-\varepsilon)$-inapproximability result for $\SASM{}$. We remark that our framework allows any future strong inapproximability results for valuation classes such as $q$-partitioning valuations~\cite{BangachevW25} to be lifted to a hardness result with additional single-minded bidders.

\paragraph{Algorithms.}

For any $\alpha$-approximable valuation class $\calV$, there is a trivial $(1+\alpha)$-approximation for \VSC{}: optimize the welfare of the succinct bidders and $\alpha$-approximate the $\calV$-type bidders separately, and choose the higher welfare allocation. It turns out that if $\calV$-type valuations are strongly $(\alpha-\varepsilon)$-inapproximable, this trivial algorithm is asymptotically optimal. Indeed, because subadditive valuations are both $2$-approximable and strongly $(2-\varepsilon)$-inapproximable, the trivial $3$-approximation algorithm for \SASC{} is optimal.

The \XOSSC{} algorithm is more interesting. The optimal $e/(e-1)$-approximation algorithm for XOS valuations solves an LP and rounds the solution using (correlated) \emph{contention resolution schemes} (CRS) for each item. Omitting details, the $e/(e-1)$-approximation stems from the fact that the optimal CRS is \emph{$(e-1)/e$-selectable}. Importantly, while the rounding procedure introduces correlations between the items, approximation guarantees for XOS valuations are unaffected by such correlations.

On the other hand, succinct bidders can be \emph{very} affected by correlations. For example, single-minded bidders get nonzero value only when receiving their entire desired set. To handle this, our algorithm first learns the full valuations of the succinct bidders, which allows them to be treated as a single surrogate bidder. Then our algorithm rounds the LP solution using \emph{online contention resolution schemes} (OCRS) for each item. The benefit of using an OCRS is that we can correlate them into all making the same decision for the first bidder. Then by presenting the surrogate succinct bidder to the OCRSs first, we avoid any unfavorable correlations. We again omit the details, but the $2$-approximation we get for \XOSSC{} stems from the fact that the optimal OCRS is $1/2$-selectable. Thus, the seemingly unrelated optimal approximation ratios of $e/(e-1)$ without succinct bidders and $2$ with succinct bidders actually stem from a deeper connection to contention resolution.

One final note is that the \XOSSC{} (\SASC{}, respectively) algorithm works for XOS (subadditive, respectively) valuations, plus a single general monotone bidder, which succinct bidders can be reduced to. On the other hand, our lower bounds only require single-minded bidders, which are a narrower class than succinct bidders. Thus, there is an interesting equivalence in the optimal approximation induced by the addition of single-minded bidders, succinct bidders, and a single general monotone bidder.

\subsection{Related Work}

As discussed previously, there has been extensive study in the communication complexity of algorithms for combinatorial auctions under a variety of valuation classes~\cite{LehmannOS02,DobzinskiNS06,BlumrosenN05a,Feige09,EzraFNTW19,DobzinskiNS10,DobzinskiV13,FeigeV10}. Our work seeks to understand the communication complexity of algorithms for combinatorial auctions with additional succinct bidders, which is directly motivated by applications in~\cite{RonTWZ24} (and to a lesser extent, applications in~\cite{QiuW24}) to lower bounds for deterministic truthful mechanisms~\cite{AssadiKSW20,Dobzinski16b}.

There is also significant work in many adjacent settings. For example, while ``truthful'' is usually shorthand for ``it is an ex-post Nash equilibrium for bidders to behave truthfully,''~\cite{DobzinskiRV22,RubinsteinSTWZ21} considers the communication complexity of \emph{dominant strategy truthful} mechanisms. There is also a long line of work in developing \emph{randomized} truthful mechanisms~\cite{DobzinskiNS06,Dobzinski07,KrystaV12,Dobzinski16a,AssadiS19,AssadiKS21}. Additionally, there is much study on the \emph{computational} complexity of truthful mechanisms, both deterministic and randomized~\cite{Dobzinski11,DughmiV11,DobzinskiV12a,DobzinskiV12b,DobzinskiV16,QiuW24}.

\section{Preliminaries}
\label{sec:prelim}

\paragraph{Combinatorial Auctions.}

In a combinatorial auction, there are a set of items $M \coloneqq [m]$ and bidders $N \coloneqq [n]$. Each bidder $i \in N$ holds a private valuation function $v_i : 2^M \to \mathbb{R}_+$ that is monotone, i.e., $v_i(S) \leq v_i(T)$ for all $S \subseteq T$. We often restrict the valuations to a valuation class $\calV$ depending on the problem setting. An \emph{allocation} is an $n$-tuple of pairwise disjoint subsets of $M$, and we denote the set of all allocations by $\Sigma$. The objective is to find an allocation $\vb{A} = (A_1, \dots, A_n) \in \Sigma$ that approximately maximizes the social welfare $\sum_{i \in N} v_i(A_i)$. Unless otherwise noted, we take the number of bidders $n \leq \poly(m)$. For notational convenience, we may also define $n$ to be the number of non-succinct bidders.

\paragraph{Communication Complexity and Protocols.}

We study the combinatorial auction problem in the general communication model. In particular, we consider the $n$-player \emph{number-in-hand blackboard communication model}, where the input known to each player $i \in N$ is her valuation function $v_i \in \calV$, and the messages sent are visible to every bidder.
The \emph{communication cost} of a protocol is the maximum expected number of bits that are written to the blackboard in the worst case, and the \emph{communication complexity} of a problem is the minimum communication cost of a protocol that solves it w.p.\ at least $2/3$.

\paragraph{Approximation.}

An allocation algorithm $\vb{A} : \calV^N \to \Sigma$ takes a valuation profile $\vb{v} = (v_1, \dots, v_n)$ as input and outputs a (possibly randomized) allocation $\vb{A}(\vb{v}) =(A_1(\vb{v}), \dots, A_n(\vb{v})) \in \Sigma$.

\begin{definition}[$\alpha$-Approximation]
    An allocation algorithm $\vb{A} : \calV^N \to \Sigma$ is an \emph{$\alpha$-approximation for $\calV$} if for every valuation profile $\vb{v} \in \calV^N$,
    \[
        \E\bigg[\sum_{i \in N} v_i(A_i(\vb{v}))\bigg] \quad \ge \quad \frac{1}{\alpha} \cdot \max_{\vb{A}^* \in \Sigma} \sum_{i \in N} v_i(A_i^*) \enspace .
    \]
\end{definition}

\begin{definition}[$\alpha$-Approximation]
    An allocation $\vb{A} \in \Sigma$ is an \emph{$\alpha$-approximation for $\vb{v}$} if
    \[
        \sum_{i \in N} v_i(A_i) \quad \ge \quad \frac{1}{\alpha} \cdot \max_{\vb{A}^* \in \Sigma} \sum_{i \in N} v_i(A_i^*) \enspace .
    \]
\end{definition}

\begin{remark} \label{remark:ApproximationEquivalence}
    Any $\alpha$-approximation algorithm $\ALG$ can be turned into an algorithm that outputs an $\alpha$-approximation w.p.\ $\Omega(1)$ by running $\ALG$ for $\ceil{\alpha}$ repetitions and outputting the best allocation. Hence, a $z$ communication lower bound for outputting an $\alpha$-approximation w.p.\ $\Omega(1)$ yields a $z/\alpha$ communication lower bound for $\alpha$-approximation algorithms.
\end{remark}

\paragraph{Valuation Classes.}

We mainly consider three canonical valuation classes.

\begin{definition}[Single-Minded]
    A valuation $v_i : 2^M \to \mathbb{R}_+$ is \emph{single-minded} if there exists a weight $w \ge 0$ and a set $T \subseteq M$ such that $v(S)= w\cdot \I[S \supseteq T]$ for all $S \subseteq M$. We denote the class of all single-minded valuation functions over $M$ by $\SM=\SM_M$.
\end{definition}

\begin{definition}[Subadditive]
    A valuation $v_i : 2^M \to \mathbb{R}_+$ is \emph{subadditive} if for all $S,T \subseteq M$, we have $v(S \cup T) \le v(S) + v(T)$. We denote the class of all subadditive valuation functions over $M$ by $\SA=\SA_M$.
\end{definition}

\begin{definition}[XOS]
    A valuation $v_i : 2^M \to \mathbb{R}_+$ is \emph{XOS} if there exists a collection of additive clauses $\calC \in \mathbb{R}^M_+$ such that $v(S)= \max_{C \in \calC} \sum_{i \in S} c_i$ for all $S \subseteq M$. We denote the class of all XOS valuation functions over $M$ by $\XOS=\XOS_M$.
\end{definition}

Note that $\XOS_M \subseteq \SA_M$. Combinatorial auctions with these classes are extensively studied, and we know communication-efficient protocols that give an asymptotically optimal $2$-approximation for \SA{}~\cite{Feige09} and asymptotically optimal $1/(1 - (1 - 1/n)^n)$-approximation for \XOS{}~\cite{DobzinskiNS10}.

\paragraph{Succinct Valuations.}

We call a valuation class $\calU$ \emph{succinct} if $\abs{\calU} \leq 2^{\poly(m)}$. In plain terms, a succinct valuation class is one whose member valuations can be ``represented'' in polynomial communication. Note that any succinct valuation class has a polynomial communication welfare maximizer.

Since representability is always relative to some communication protocol, the notion of ``succinct valuations'' is not a well-defined concept, though we will informally use the term ``succinct valuation'' and refer to \SC{} as the ``class of succinct valuations'' for convenience. Formally, for any valuation class $\calV$, we say that
\begin{itemize}[topsep=4pt]
    \item There exists an $\alpha$-approximation for \VSC{} in $z$ communication if for every succinct valuation class $\calU$, there exists an $\alpha$-approximation algorithm for $\calV\,\cup\,\calU$ in $z$ communication.

    \item Any $\alpha$-approximation for \VSC{} uses $z$ communication if for some succinct valuation class $\calU$, any $\alpha$-approximation algorithm for $\calV\,\cup\,\calU$ uses $z$ communication.
\end{itemize}

\section{Algorithms}
\label{sec:algo}

In this section, we design communication-efficient algorithms for \SASC{} and \XOSSC{}.

\begin{theorem} \label{thm:3-approx-alg-for-SASM}
    There exists a $(3-2/n)$-approximation for \SASC{} in $\poly(m,n)$ communication.
\end{theorem}

\begin{theorem} \label{thm:2-approx-alg-for-XOSSM}
    There exists a $2$-approximation for \XOSSC{} in $\poly(m,n)$ communication.
\end{theorem}

\subsection{\texorpdfstring{$(3-2/n)$}{(3-2/n)}-approximation for \texorpdfstring{\SASC{}}{SA+Succ}}

The $(3-2/n)$-approximation for \SASC{} is straightforward: since we can optimize the welfare of just the succinct bidders, and we can $2$-approximate the welfare from just the subadditive bidders, we can $3$-approximate the overall welfare by just choosing the better of the two approaches.

\begin{proof}[Proof of \Cref{thm:3-approx-alg-for-SASM}]
For any subadditive bidder $i$, we can find the optimal welfare $\OPT_i$ among all the succinct bidders and subadditive bidder $i$: have the succinct bidders send their valuations to the subadditive bidder, who now knows all the relevant valuations and can compute $\OPT_i$.

If the succinct bidders constitute at least a $(n-2)/(3n-2)$ fraction of the welfare in the optimal allocation, then some $\OPT_i$ achieves at least $(n-2)/(3n-2) + (1 - (n-2)/(3n-2))/n = n/(3n-2)$ fraction of the optimal welfare. Otherwise, the subadditive bidders constitute at least $2n/(3n-2)$ fraction of the welfare in the optimal allocation, of which we can get half using~\Cref{thm:ApproxForSAAndXOS}~\cite{Feige09}. Thus, this algorithm guarantees a $(3n-2)/n = (3 - 2/n)$-approximation.

The algorithm only has the succinct bidders send their valuations and the subadditive bidders run the $2$-approximation algorithm of~\cite{Feige09}, both of which use polynomial communication.
\end{proof}

\subsection{\texorpdfstring{$2$}{2}-approximation for \texorpdfstring{\XOSSC{}}{XOS+Succ}}

To obtain the $2$-approximation for \XOSSC{} valuations, we start by grouping the succinct bidders into a single surrogate bidder with a general monotone valuation, solving the configuration LP for this bidder and the XOS bidders to obtain a fractional solution.
Then, we make use of \emph{online contention resolution schemes} (OCRS) for rank $1$ matroids to round the fractional solution in a correlated way.
The correlation ensures the succinct bidders will always get either the entirety of their desired bundle or nothing, without taking away too much from the XOS bidders.

Let $N_{\SC}$ be the set of succinct bidders and $N_{\XOS}$ be the set of XOS bidders (so $N_{\XOS} \sqcup N_{\SC} = N$).
We first create a surrogate bidder $0$ with valuation $v_0$ defined by
\begin{equation}\label{eq:surrogate-valuation}
  v_0(S) \; = \; \max_{\vb{A} \in \Sigma(S)} \sum_{i \in N_\SM} v_i(A_i) \qquad \forall S \subseteq M \enspace .  
\end{equation}
Here, $\Sigma(S)$ is the collection of all feasible allocations $\vb{A}$ where $A_i \in S$ for all $i \in N_{\SM}$, and $A_i \cap A_j = \emptyset$ for all $i \ne j$.
We want to solve the following configuration LP for $N'=N_\XOS \cup \{0\}$ (i.e., the XOS bidders and the surrogate bidder $0$):
\begin{equation} \label{eq:configuration-lp}
\begin{aligned}
\max\;& \sum_{i \in N'}\sum_{S\subseteq M} v_i(S)\,x_{i,S}\\
\text{s.t.}\;& \sum_{i \in N'}\sum_{S\ni j} x_{i,S}\le 1 &&\forall j\in M \, , \\
& \sum_{S\subseteq M} x_{i,S}\le 1 &&\forall i\in N' \, , \\
& x_{i,S}\ge 0 &&\forall i\in N'\,,\;S\subseteq M \, .
\end{aligned}
\end{equation}
Even though the LP has exponentially many variables, classical results \cite{NisanS06, BlumrosenN05a} show that it can be solved in $\poly(m,n)$ communication by solving its dual using demand queries as a separation oracle.
The most straightforward way to round the fractional solution is to choose a set $S_i$ w.p. $x_{i,S}$ for each bidder $i \in N'$, but this can result in an infeasible allocation.

Instead, our algorithm will make use of the $1/2$-OCRS for rank 1 matroids~\cite{FeldmanSZ16, Alaei11, LeeS18} to round the fractional solution and obtain an $1/2$-approximation.

\begin{definition}[$\alpha$-OCRS~\cite{FeldmanSZ16} for Rank $1$ Matroids]
    Let $\vb{p} \in [0, 1]^k$ be such that $\norm{\vb{p}}_1 \leq 1$, and let $X_1, X_2, \ldots, X_k$ be independent Bernoulli random variables such that $\Pr[X_i = 1] = p_i$ for all $i \in [k]$. 
    An \emph{online contention resolution scheme} (OCRS) is an algorithm which in each time step $i \in [k]$ sees the value of $X_i$. If $X_i = 0$, nothing happens. If $X_i = 1$, the algorithm can choose to accept or reject $X_i$, subject to only accepting at most once. An $\alpha$-OCRS is an OCRS which accepts each $X_i$ w.p.\ at least $\alpha$, conditioned on $X_i = 1$.
\end{definition}

\begin{proposition}[\cite{Alaei11, LeeS18}]
    There exists a $1/2$-OCRS for rank 1 matroids. Additionally, the OCRS accepts the first element $X_1$ w.p.\ $1/2$ whenever $X_1 = 1$.
\end{proposition}

Now we are ready to formally state our algorithm:
\begin{enumerate}[topsep=4pt]
    \item Let $N_\XOS$ be the set of XOS bidders and $N_\SC$ be the set of succinct bidders. Create a surrogate bidder~$0$ with valuation function defined in~\eqref{eq:surrogate-valuation}.
    
    \item Solve the configuration LP~\eqref{eq:configuration-lp} for $N'=N_\XOS \cup \{0\}$ to obtain an optimal fractional solution $\{x_{i,S}^*\}_{i\in N',\,S\subseteq M}$.

    \item For each bidder $i\in N'$, independently draw one bundle $S_i$ according to the distribution over subsets $S \subseteq M$ where $\Pr[S_i=S]=x_{i,S}^*$.

    \item For each item $j \in M$, run the $1/2$-OCRS of~\cite{Alaei11, LeeS18} on $\{X_i\}_{i \in N'}$, where $X_i = 1$ if $j \in S_i$ and $X_i = 0$ otherwise. The OCRSs are correlated so that the decision to accept the first element $X_0$ (if $X_0 = 1$) is the same across all $m$ instances, which is possible because $X_0$ (if $X_0 = 1$) is selected w.p. exactly $1/2$.

    \item Define the sets $T_i \subseteq S_i$ for $i \in N'$ to be exactly the items $j$ for which the OCRS on item $j$ accepted $X_i$. We output the allocation $\vb{A} \in \Sigma$, where $A_i = T_i$ for all $i \in N_\XOS$, and where the items in $T_0$ are allocated among the succinct bidders in a way that maximizes their welfare.
\end{enumerate}

Our analysis of the approximation ratio relies on the following property of XOS functions over a (possibly correlated) random set.

\begin{proposition}[\cite{Feige09}] \label{prop:XOSApprox}
    Let $v : 2^M \to \mathbb{R}_+$ be XOS. For any $T \subseteq M$ and $p \in [0, 1]$, let $S$ be a random set where for each $j \in T$, $\Pr[j \in S] \geq p$, not necessarily independently. Then $\E[v(S)] \geq p \cdot v(T)$.
\end{proposition}

\begin{proof}[Proof of \Cref{thm:2-approx-alg-for-XOSSM}]
Observe that for every allocation $\vb{A} \in \Sigma$, the LP~\eqref{eq:configuration-lp} solution
\[
    x_{i,S} \quad = \quad \begin{cases}
        1 & i \in N_\XOS, S = A_i \\
        1 & i = 0, S = \bigcup_{i' \in N_{\SC}} A_{i'} \\
        0 & o.w.
    \end{cases}
\]
is feasible and has at least as much value (more value, if the succinct bidders are allocated suboptimally in $\vb{A}$), so the value $\OPT$ of the LP upper bounds the optimal welfare.

Let $\OPT_i$ be the value that bidder $i \in N'$ contributes to $\OPT$, i.e., $\OPT_i = \sum_{S \subseteq M} x_{i,S}^* v_i(S)$.

\begin{lemma}
    For all $i \in N'$, $\E[v_i(T_i)] \geq \OPT_i/2$.
\end{lemma}
\begin{proof}
$\E[v_0(T_0)] = \OPT_0/2$ because we have correlated the $1/2$-OCRS for each item to set $T_0 = S_0$ w.p.\ exactly $1/2$. Additionally, the condition $\norm{p}_1 \leq 1$ of the $1/2$-OCRS for rank 1 matroid is satisfied due to the item constraint in the configuration LP, so when XOS bidder $i \in N_\XOS$ draws the bundle $S_i$ of each XOS bidder $i \in N_\XOS$, they receive each item $j \in S_i$ w.p.\ at least $1/2$ by the guarantee of the $1/2$-OCRS. Therefore, by \Cref{prop:XOSApprox}, $\E[v_i(T_i)] \geq \OPT_i/2$ for all $i \in N_\XOS$.
\end{proof}

The allocation $\vb{A} \in \Sigma$ output by the algorithm has welfare at least $\sum_{i \in N'} v_i(T_i)$ because each bidder $i \in N_\XOS$ receives $T_i$ exactly, and the succinct bidders are given the allocation of the items in $T_0$ which gives welfare precisely $v_0(T_0)$. Thus, the algorithm gives a $2$-approximation for \XOSSC{}.

Finally, observe that the algorithm only requires the bidders to solve the configuration LP, after which a designated bidder can run the rounding procedure by herself. Since $v_0$ can be specified in polynomial communication (only need each succinct bidder to send their valuation), any query to $v_0$ can be simulated in polynomial communication. Thus, the configuration LP can be solved (and hence the algorithm can be run) in polynomial communication.
\end{proof}

\section{Hardness of Approximation}
\label{sec:hardness}

Let $\calV$ be a valuation class closed under addition and point-wise maximum (i.e., for $f, g \in \calV$, $f + g \in \calV$ and $\max\{f, g\} \in \calV$). Note that \XOS{} is the minimal non-trivial valuation class satisfying these properties.

We now define an important property and state the main theorem.

\begin{definition}[name=Scarce,restate=Scarce]
    An allocation is \emph{$t$-scarce} if all but at most $t$ bidders are allocated $\emptyset$.
\end{definition}

\begin{definition}[Strong Inapproximability (Informal)]
    Let $\alpha : [n] \to \mathbb{R}_+$ be a function. A valuation class is \emph{strongly $\alpha(\cdot)$-inapproximable in $z$ communication} if for all $t \in [n]$, beating an $(\alpha(t) \cdot n/t)$-approximation with a $t$-scarce allocation requires $z$ communication.
\end{definition}

\begin{theorem} \label{thm:HardnessOfApproximation}
    Let $\calV$ be a valuation class closed under addition and point-wise maximum, and let there be $n$ $\calV$-type bidders and $c$ single-minded bidders. For any $\lambda \geq 1$ let $\alpha_\lambda(t) \coloneqq \lambda(1 - 1/(t+1))$. Fix any $a \leq n$. Let $r$ be such that in an instance with $r$ items, $\calV$ is strongly $\alpha_\lambda(\cdot)$-inapproximable in $2^{\Omega(\min\{c, m/r\}/a)}$ communication. Then any $(1 + \lambda)(1 - 4a^{-1/3})$-approximation for \VSM{} uses $2^{\Omega(\min\{c, m/r\}/a)}$ communication.
\end{theorem}

We first outline the main ideas of the proof. For simplicity, let $a = n$ (otherwise, discard the additional bidders) and let $m = rc$ (otherwise, discard the additional bidders/items). The hard instance for \VSM{} is constructed as follows:
\begin{itemize}[topsep=4pt]
    \item Arrange the items into $c$ columns of $r$ items $R \times C$, and let $\varepsilon > 0$ be small.

    \item Each single-minded bidder $j$ gets value $\varepsilon$ for all the items in column $j$.

    \item The $\calV$-type bidders form (the same) arbitrary instance with welfare $\lambda$ on the items within each of the columns.

    \item Each $\calV$-type bidder $i$ holds an exponentially-sized collection of column-sets $\calS_i \subseteq C$, where every $S \in \calS_i$ is of size $\varepsilon c$. Across columns, each $\calV$-type bidder $i$ adds up their values in the column-set $S \in \calS_i$ which maximizes their value (this is where we need $\calV$ to be closed under addition and point-wise maximum).

    \item There exists a single set $S^* \in \calS_1 \cap \dots \cap \calS_n$.
\end{itemize}

The optimal allocation gives columns $S^*$ to the $\calV$-type bidders and the remaining columns to the single-minded bidders, which yields welfare $(1 + \lambda - \varepsilon)\varepsilon c$.

However, since $\calS_1, \dots, \calS_n$ are exponentially-sized, an efficient algorithm cannot hope to find $S^*$. Further, by (roughly) $\lambda$-inapproximability, an efficient algorithm cannot obtain more than $1/\lambda$ of the welfare from the $\calV$-type bidders. So giving everything to the single-minded bidders yields welfare $\varepsilon c$, and giving everything to the $\calV$-type bidders (via a $\lambda$-approximation algorithm) also yields welfare $\varepsilon c$.

There are two additional technical considerations:
\begin{enumerate}[topsep=4pt]
    \item Though an algorithm may not be able to find $S^*$ exactly, if it can find a set of columns $S$ which nonetheless is able to satisfy many of the $\calV$-type bidders, then we can eke out additional value by giving columns $S$ to the $\calV$-type bidders and giving everything else to the single-minded bidders. For example, if $n \ll 1/\varepsilon$, we can easily find a set of $\varepsilon nc \ll c$ columns (without learning anything about $S^*$) which allow the $\calV$-type bidders to achieve full welfare while still being able to satisfy most of the single-minded bidders.

    \item Since each $\calV$-type bidder only gains value from $\varepsilon c$ columns at a time, an algorithm can split them up so that each column is only ``responsible'' for $\varepsilon n$ of the $\calV$-type bidders. We therefore need the (roughly) $\lambda$-inapproximability of the instance to hold even when only a small fraction of the bidders are being considered.
\end{enumerate}

We define the notion of strong inapproximability precisely to handle $(2)$. The two results we will use are given below, the former of which is proved in~\Cref{sec:inapproximable}.

\begin{proposition}[name=,restate=SubadditiveStronglyInapproximableProp] \label{prop:SubadditiveStronglyInapproximable}
    \SA{} is strongly $(2-2/(t+1))$-inapproximable in $2^{\Omega(m/n^2)}$ communication.
\end{proposition}

\begin{proposition} \label{prop:XOSStronglyInapproximable}
    \XOS{} is strongly $1$-inapproximable when $m \geq n$.
\end{proposition}
\begin{proof}
Consider the instance where bidder $i$ has marginal value $1$ for item $i$ and $0$ for everything else. Then the optimal welfare is $n$, and no $t$-scarce allocation can have welfare more than $t$.
\end{proof}

The following lemma handles $(1)$, formalizing the meaning of ``without uncovering $S^*$, it is impossible to find good $S$ for the $\calV$-type bidders.''

\begin{lemma} \label{lemma:HardToGatherSets}
    Let $\varepsilon = n^{-1/3}$ and $z = 2^{\Theta(c/n)}$. There exist sets $S_1, \dots, S_z \subseteq C$ of size $\varepsilon c$ such that for all $\delta \in [0, 1]$, all $S \subseteq C$ of size $\delta c$, and all $L \subseteq [z]$ of size $n$,
    \[
        \sum_{\ell \in L} \abs{S \cap S_\ell} \quad \leq \quad (\delta + 3\varepsilon) \varepsilon nc \enspace .
    \]
\end{lemma}
\begin{proof}
Let $T_1, \dots, T_z \subseteq C$ be sets sampled by independently including each item w.p. $\varepsilon/(1-\varepsilon/4)$. If $\abs{T_\ell} < \varepsilon c$ for any $\ell \in [z]$, we abort. Otherwise, we let each $S_\ell$ be an arbitrary subset of $T_\ell$ of size $\varepsilon c$. Observe that the lemma follows for $S_1, \dots, S_z$ so long as it follows for $T_1, \dots, T_z$.

Fix $\delta \in [0, 1]$, $S \subseteq C$ of size $\delta c$, and $L \subseteq [z]$ of size $n$. Observe that $Z_{S, L} \coloneqq \sum_{\ell \in L} \abs{S \cap T_\ell} \sim \Binom(\delta nc, \varepsilon/(1-\varepsilon/4))$, so
\begin{align*}
    \Pr[Z_{S, L} \geq (\delta + 3\varepsilon)\varepsilon nc] \quad &\leq \quad \Pr[Z_{S, L} \geq \E[Z_{S, L}] + 2\varepsilon\E[Z_{C, L}]] \\
    &\leq \quad \exp\bigg(-\bigg(2\varepsilon\frac{\E[Z_{C, L}]}{\E[Z_{S, L}]}\bigg)^2 \frac{\E[Z_{S, L}]}{2+\varepsilon}\bigg) \\
    &\leq \quad \exp\bigg(-\frac{4\varepsilon^2}{2+\varepsilon}\E[Z_{C, L}]\bigg) \quad \leq \quad \exp(-\varepsilon^3 nc) \quad = \quad e^{-c} \enspace .
\end{align*}
Taking a union bound over all $2^c$ choices of $S$ and all at most $z^n = 2^{\Theta(c)}$ choices of $L$, we see that $Z_{S, L} \leq (\delta + 3\varepsilon)\varepsilon ny$ for all $S, L$ simultaneously w.h.p.

Additionally, by a Chernoff bound and union bound, the probability that $\abs{T_\ell} < \varepsilon c$ for some $\ell \in [z]$ is at most $z e^{-\Theta(\varepsilon^3 c)} \ll 1$, so $S_1, \dots, S_z$ satisfies the lemma w.h.p.
\end{proof}

Before we prove the main result, we will need some definitions and a fact (see its proof in~\Cref{sec:missing-proofs}) regarding communication protocols.

\begin{definition}[\textsc{PromiseDisjointness}]
    In \textsc{PromiseDisjointness}, each player receives as input a set $X_i \subseteq [z]$. The players must output
    \begin{itemize}[topsep=4pt]
        \item $1$, if there exists $\ell \in \bigcap_{i \in [n]} X_i$.

        \item $0$, if $X_1, \dots, X_n$ are pairwise disjoint.

        \item Anything, if neither of the above is true.
    \end{itemize}
    Solving \textsc{PromiseDisjointness} w.p. $2/3$ requires $\Omega(z/n)$ communication \cite{Gro09}.
\end{definition}

\begin{definition}[\textsc{$(\calV, \alpha(\cdot))$-GapWelfare}]
    In \textsc{$(\calV, \alpha(\cdot))$-GapWelfare}, each player receives a $\calV$-type valuation $f_i : 2^R \to \mathbb{R}_+$. The players must output
    \begin{itemize}[topsep=4pt]
        \item $1$, if the welfare of the instance $f_1, \dots, f_n$ is $n$.

        \item $0$, if for all $t \in [n]$, there does not exist a $t$-scarce allocation with welfare greater than $t/\alpha(t)$.

        \item Anything, if neither of the above is true.
    \end{itemize}
\end{definition}

\begin{definition}[name=Strong Inapproximability,restate=StronglyInapproximable] \label{def:StronglyInapproximable}
    A valuation class $\calV$ is \emph{strongly $\alpha(\cdot)$-inapproximable in $z$ communication} if solving \textsc{$(\calV, \alpha(\cdot))$-GapWelfare} w.p.\ $2/3$ requires $z$ communication.
\end{definition}

\begin{proposition}[name=,restate=TwoProblems] \label{lemma:TwoProblems}
    Let $\calX, \calY$ be some domains, let $\calX^* \subseteq \calX^n, \calY^* \subseteq \calY^n$, and let $f : \calX^n \to \{0, 1\}$ and $g : \calY^n \to \{0, 1\}$. Let $\calQ_f$ be the $n$-player communication problem where the players receive as input $\vb{x} \in \calX^n$, and must output $f(\vb{x})$ if $\vb{x} \in \calX^*$ (and may output anything otherwise), and let $\calQ_g$ be the $n$-player communication problem where the players receive as input $\vb{y} \in \calY^n$, and must output $g(\vb{y})$ if $\vb{y} \in \calY^*$ (and may output anything otherwise). Let $\calQ$ be the $n$-player communication problem where the players receive as input $(\vb{x}, \vb{y}) \in (\calX \times \calY)^n$, and must output $f(\vb{x})$ if $f(\vb{x}) = g(\vb{y})$, and may output anything if $f(\vb{x}) \ne g(\vb{y})$. Then if $\calQ_f$ and $\calQ_g$ require $\Omega(z)$ communication to solve w.p. $2/3$, $\calQ$ also requires $\Omega(z)$ communication to solve w.p. $2/3$.
\end{proposition}

\begin{proof}[Proof of~\Cref{thm:HardnessOfApproximation}]
WLOG, let $a = n$ and $m = rc$ as before (otherwise, discard the additional bidders/items), and label the items $M \coloneqq R \times C$. Let $\varepsilon = n^{-1/3}$ and let $z = 2^{\Theta(c/n)}$. The theorem statement then becomes, ``any $(1+\lambda)(1-4\varepsilon)$-approximation for \VSM{} uses $z$ communication.''

Let $S_1, \dots, S_z \subseteq C$ be sets satisfying the property of \Cref{lemma:HardToGatherSets}. Let $X_1, \dots, X_n \subseteq [z]$ be a \textsc{PromiseDisjointness} instance, and let $f_1, \dots, f_n : 2^R \to \mathbb{R}_+$ be a \textsc{$(\calV, \alpha_\lambda(\cdot))$-GapWelfare} instance, such that both are $1$-instances or both are $0$-instances.

Let $u_1, \dots, u_c$ be single-minded bidders defined $u_j(S) \coloneqq \varepsilon n \cdot \I(S \supseteq (R \times \{j\}))$. Let $v_1, \dots, v_n : 2^M \to \mathbb{R}_+$ be defined
\[
    v_i(S) \quad \coloneqq \quad \lambda\max_{\ell \in X_i} \sum_{j \in S_\ell} f_i(S \cap (R \times \{j\})) \enspace ,
\]
and observe that $v_i \in \calV$ because $\calV$ is closed under addition and point-wise maximum. In plain terms, $v_i$ is the valuation which participates in a $\calV$-type instance in each column, and adds its values across columns in an XOS manner.

If both communication problems are $1$-instances, there exists $\ell \in \bigcap_{i \in [n]} X_i$, and the optimal welfare of $f_1, \dots, f_n$ is $n$. Then allocating $R \times S_\ell$ optimally among bidders $v_1, \dots, v_n$ and allocating the remaining columns to $u_1, \dots, u_c$ yields $\varepsilon c \cdot \lambda n + (1 - \varepsilon) c \cdot \varepsilon n = (1 + \lambda - \varepsilon) \varepsilon nc$.

If both communication problems are $0$-instances, consider any set of columns $J \subseteq C$ of size $\delta c$ given to $v_1, \dots, v_n$. For $T \subseteq [n]$, let $\OPT(T)$ denote the optimal welfare among bidders $T$ in a single column, and observe that $\OPT(T) \leq \lambda\abs{T}/(\lambda(1-1/(\abs{T}+1))) = \abs{T} + 1$ because we are in a $0$-instance of \textsc{$(\calV, \alpha_\lambda(\cdot))$-GapWelfare}. Then denoting $\vb{\ell} = (\ell_1, \dots, \ell_n)$ and $\vb{X} = (X_1, \dots, X_n)$, the optimal welfare among $v_1, \dots, v_n$ is
\begin{align*}
    \max_{\vb{\ell} \in \vb{X}} \sum_{j \in J} \OPT(\{i \in [n] : j \in S_{\ell_i}\}) \quad &\leq \quad \max_{\vb{\ell} \in \vb{X}} \sum_{j \in J} (\abs{\{i \in [n] : j \in S_{\ell_i}\}} + 1) \\
    &\leq \quad c + \max_{\vb{\ell} \in \vb{X}} \sum_{i \in [n]} \abs{J \cap S_{\ell_i}} \\
    &\leq \quad c + \max_{\substack{L \subseteq [z] \\ \abs{L} = n}} \sum_{\ell \in L} \abs{J \cap S_\ell} \\
    &\leq \quad (\delta + 3\varepsilon)\varepsilon nc \\
    &\leq \quad (\delta + 3\varepsilon + \varepsilon^2)\varepsilon nc \enspace ,
\end{align*}
where the third line comes from the fact that $\vb{\ell}$ has distinct coordinates due to the disjointness of $X_1, \dots, X_n$, and the fourth line comes from \Cref{lemma:HardToGatherSets}. Thus, the total welfare cannot exceed $(1 - \delta)\varepsilon nc + (\delta + 3\varepsilon + \varepsilon^2)\varepsilon nc = (1 + 3\varepsilon + \varepsilon^2)\varepsilon nc$.

Therefore, beating a $(1 + \lambda - \varepsilon)/(1 + 3\varepsilon + \varepsilon^2) \geq (1 + \lambda)(1 - 4\varepsilon)$-approximation distinguishes between the cases where for instances \textsc{PromiseDisjointness} and \textsc{$(\calV, \alpha_\lambda(\cdot))$-GapWelfare}, both are $1$-instances or both are $0$-instances. Since \textsc{PromiseDisjointness} has communication complexity $\Omega(z/n)$ and \textsc{$(\calV, \alpha_\lambda(\cdot))$-GapWelfare} has communication complexity $z$, this requires $\Omega(z/n) = 2^{\Omega(c/n)}$ communication by~\Cref{lemma:TwoProblems}.\footnote{Note that the precise statement we have proved is of the form ``any algorithm which finds an $\alpha$-approximation w.p.\ $2/3$ uses $z$ communication,'' but this implies ``any $\alpha$-approximation algorithm uses $z$ communication'' by~\Cref{remark:ApproximationEquivalence}.}
\end{proof}

\ComplexitySASM*
\begin{proof}
The lower bound follows by \Cref{prop:SubadditiveStronglyInapproximable} and \Cref{thm:HardnessOfApproximation} with $r = \sqrt{ma}$, and the upper bound follows by~\Cref{thm:3-approx-alg-for-SASM}.
\end{proof}

\ComplexityXOSSM*
\begin{proof}
The lower bound follows by \Cref{prop:XOSStronglyInapproximable} and \Cref{thm:HardnessOfApproximation} with $r = a$, and the upper bound follows by~\Cref{thm:2-approx-alg-for-XOSSM}.
\end{proof}

\begin{remark}
    The polynomial communication $2$-approximation for \XOSSM{} implies that \XOS{} is not strongly $(1 + \varepsilon)$-inapproximable for any constant $\varepsilon > 0$.
\end{remark}

\section{Separations with Additional Succinct Bidders}
\label{sec:separations}

In the previous section, we gave the correct hardness of approximations for \SASM{} and \XOSSM{} of $3$ and $2$, respectively, as the number of non-succinct bidders $n \to \infty$. In this section, we demonstrate that for all (non-trivial) $n$, there is a constant gap in the hardness of approximations between \SASM{} and \SA{}, and between \XOSSM{} and \XOS{}.

Observe that there are optimal (up to lower order terms) trivial algorithms for \XOS{} when $n = 1$ and \SA{} when $n = 2$; simply give everything to the higher value bidder. These algorithms are easily extended to trivial algorithms for \XOSSM{} and \SASM{} which achieve the same approximation ratio: ask each bidder to maximize the welfare among the succinct bidders and herself, and choose the highest welfare allocation. We call these values of $n$ trivial for hardness of approximation results because the single-minded bidders can only improve the approximation ratio; the only interesting values of $n$ to consider are $n \geq 3$ for \SASM{}, and $n \geq 2$ for \XOSSM{}.

Our constructions for the separation results are similar to those for the asymptotic hardness of approximation results: we divide the items into columns $C \coloneqq [c]$, with a single-minded bidder interested in each column. However, since $n$ can now be very small, we cannot force the optimal allocation to give the non-succinct bidders a very small column set, because then we would have enough columns to give each non-succinct bidder a disjoint column set. Therefore, we need the non-succinct bidders to receive the majority of columns in the optimal allocation. In fact, because XOS valuations are not hard to approximate unless most of the items are under contention, we need each XOS bidder to be interested very large column sets.

We will need the following two strong inapproximability results, the former of which is proved in~\Cref{sec:inapproximable}, as well as a probabilistic lemma to prove our separations.

\SubadditiveStronglyInapproximableThm*

\begin{proposition} \label{prop:XOSStronglyInapproximable2}
    For constant $n$, \XOS{} is strongly $((t/n)/(1 - (1 - 1/n)^t) - o(1))$-inapproximable in $2^{\Omega(\sqrt{m})}$ communication.
\end{proposition}
\begin{proof}
For $z = 2^{\Theta(\sqrt{m})}$, sample uniformly random partitions $\vb{A}_1, \dots, \vb{A}_z \in \Sigma$. For all $T \subseteq N$ and $L \subseteq [z]$ of size $n$, we have $Z_{T,L} \coloneqq \abs{\bigcup_{i \in T} A_{\ell_i,i}} \sim \Binom(m, 1 - (1 - 1/n)^t)$, so by a Chernoff bound,
\[
    \Pr[Z_{T,L} \geq (1-(1-1/n)^t)m + m^{3/4}] \quad = \quad 2^{-\Omega(\sqrt{m})} \enspace .
\]
By a union bound over all at most $2^n z^n = 2^{\Theta(\sqrt{m})}$ choices of $T, L$, there exist partitions $\vb{A}_1, \dots, \vb{A}_z \in \Sigma$ for which $Z_{T,L} \leq (1-(1-1/n)^t)m + m^{3/4}$ for all $T \subseteq N$ and $L \subseteq [z]$ of size $n$. Fix any such partitions.

Consider a \textsc{PromiseDisjointness} instance with inputs $X_1, \dots, X_n \subseteq [z]$. For $v_1, \dots, v_n : 2^M \to \mathbb{R}_+$, define $v_i(S) \coloneqq \max_{\ell \in X_i} \abs{S \cap A_{\ell,i}}$, and observe that $v_i$ is XOS.

Suppose there exists $\ell \in \bigcap_{i \in N} X_i$. Then the allocation $\vb{A}_\ell$ yields $m$ welfare.

Suppose $X_1, \dots, X_n$ are pairwise disjoint, and consider any $t$-scarce allocation $\vb{A}$ on bidders $T \subseteq N$. Then by disjointness of $X_1, \dots, X_n$ and the above property of $\vb{A}_1, \dots, \vb{A}_z$,
\[
    \sum_{i \in T} \max_{\ell_i \in X_i} \abs{A_i \cap A_{\ell_i,i}} \quad \leq \quad \max_{\vb{\ell} \in \vb{X}} \bigg\vert \bigcup_{i \in T} A_{\ell_i,i} \bigg\vert \quad \leq \quad (1 - (1 - 1/n)^t)m + m^{3/4} \enspace .
\]

Therefore, any algorithm which decides if there exists a $t$-scarce allocation with more than a $((1 - (1 - 1/n)^t)m + m^{3/4})/m = 1 - (1 - 1/n)^t + o(1)$ fraction of the optimal welfare is able to solve \textsc{PromiseDisjointness} on $[z]$, which requires $\Omega(z/n) = 2^{\Omega(\sqrt{m})}$ communication.
\end{proof}

\begin{lemma} \label{lemma:WhatWeExpect}
    Let $n$ and $p, \delta \in (0, 1)$ be constants. Then for $z = 2^{\Theta(c)}$, there exist sets $S_1, \dots, S_z \subseteq C$ of size in $[(1 - \delta)pc, (1 + \delta)pc]$ such that for any disjoint $K, L \subseteq [z]$ of combined size $n$,
    \[
        \bigg\vert \bigg(\bigcap_{\ell \in K} \closure{S}_\ell\bigg) \cap \bigg(\bigcap_{\ell \in L} S_\ell\bigg) \bigg\vert \quad \leq \quad (1 + \delta)(1-p)^{\abs{K}} p^{\abs{L}} c \enspace .
    \]
\end{lemma}
\begin{proof}
Let $S_1, \dots, S_z \subseteq C$ be sets sampled by independently including each item w.p. $p$. $p, \delta = \Omega(1)$, so a union bound and a Chernoff bound show $\abs{S_\ell} \in [(1 - \delta)pc, (1 + \delta)pc]$ for all $\ell \in [z]$ w.h.p.

Now, fix disjoint $K, L \subseteq [z]$ of combined size $n$, and observe that
\begin{align*}
    Z_{K,L} \quad \coloneqq \quad \bigg\vert \bigg(\bigcap_{\ell \in K} \closure{S}_\ell\bigg) \cap \bigg(\bigcap_{\ell \in L} S_{\ell}\bigg) \bigg\vert \quad &\sim \quad \Binom(c, (1 - p)^{\abs{K}} p^{\abs{L}}) \enspace .
\end{align*}

Then by a Chernoff bound, $\Pr[Z_{K,L} \geq (1 + \delta)(1-p)^{\abs{K}} p^{\abs{L}} y] \leq e^{-\delta^2 (1-p)^{\abs{K}} p^{\abs{L}} c/3} = e^{-\Theta(c)}$. Taking a union bound over all at most $z^n = e^{\Theta(c)}$ choices of $K, L$ completes the proof.
\end{proof}

\begin{theorem}
    Let $\calV$ be a valuation class closed under addition and point-wise maximum, and let there be $n$ $\calV$-type bidders and $c$ single-minded bidders, where $n = O(1)$. Suppose that $\calV$ is strongly $\alpha(\cdot)$-inapproximable in $2^{\Omega(\sqrt{m})}$ communication. Then for any constants $p, \delta \in (0, 1)$, a $\beta(n, p, \delta)$-approximation for \VSM{} uses $\min\{2^{\Omega(\min\{c, m^{1/3}\})}\}$ communication, where (defining $0/\alpha(0) = 0$)
    \[
        \beta(n, p, \delta) \;\; \coloneqq \;\; \max_{t^* \in [n]} \bigg\{\bigg(pn + (1-p) \frac{t^*}{\alpha(t^*)}\bigg)\bigg/\bigg(\sum_{t=0}^n \binom{n}{t}(1-p)^{n-t} p^t \max\bigg\{\frac{t}{\alpha(t)},\, \frac{t^*}{\alpha(t^*)}\bigg\}\bigg)\bigg\} - \delta \enspace .
    \]
\end{theorem}
\begin{proof}
WLOG, let $m = c^3$ (otherwise, discard the additional bidders/items), and label the items $M \coloneqq R \times C$ (where $R \coloneqq [c^2])$. Let $z \coloneqq 2^{\Theta(c)}$, and fix any $t^* \in [n]$.

Let $S_1, \dots, S_z \subseteq C$ be sets satisfying the property of~\Cref{lemma:WhatWeExpect}. Let $X_1, \dots, X_n \subseteq [z]$ be a \textsc{PromiseDisjointness} instance, and let $f_1, \dots, f_n : 2^R \to \mathbb{R}_+$ be a \textsc{$(\calV, \alpha(\cdot))$-GapWelfare} instance, such that both are $1$-instances or both are $0$-instances.

Let $u_1, \dots, u_c$ be single-minded bidders defined $u_j(S) \coloneqq t^*/\alpha(t^*) \cdot \mathbbm{1}(S \supseteq (X \times \{j\})$. Let $v_1, \dots, v_n$ be defined
\[
    v_i(S) \quad \coloneqq \quad \max_{\ell \in X_i} \sum_{j \in S_\ell} f_i(S \cap (R \times \{j\})) \enspace ,
\]
and observe that $v_i \in \calV$ because $\calV$ is closed under addition and point-wise maximum.

If both communication problems are $1$-instances, there exists $\ell \in \bigcap_{i \in [n]} X_i$, and the optimal welfare of $f_1, \dots, f_n$ is $n$. Then $\abs{S_\ell} \geq (1-\delta)pc$ by \Cref{lemma:WhatWeExpect}, and allocating $R \times S_\ell$ optimally among bidders $v_1, \dots, v_n$ and allocating the remaining columns to $u_1, \dots, u_c$ yields $\abs{S_\ell} \cdot n + (c - \abs{S_\ell}) \cdot t^*/\alpha(t^*) \geq (1-\delta)(pn + (1-p) t^*/\alpha(t^*))c$ welfare.

If both communication problems are $0$-instances, let $\OPT(T)$ denote the optimal welfare among bidders $T$ in a single column, and observe that $\OPT(T) \leq \abs{T}/\alpha(\abs{T})$ because we are in a $0$-instance of \textsc{$(\calV, \alpha(\cdot))$-GapWelfare}. Then the optimal welfare is
\[
    \OPT \quad = \quad \max_{\vb{\ell} \in \vb{X}} \bigg\{\sum_{T \subseteq [n]} \bigg(\bigg\vert \bigg(\bigcap_{i \in T} S_{\ell_i}\bigg) \cap \bigg(\bigcap_{i \not\in T} \closure{S}_{\ell_i}\bigg) \bigg\vert \max\bigg\{\frac{\abs{T}}{\alpha(\abs{T})},\, \frac{t^*}{\alpha(t^*)}\bigg\}\bigg)\bigg\} \enspace .
\]

Since $X_1, \dots, X_n$ are disjoint, $\vb{\ell} \in \vb{X}$ has distinct coordinates, and hence by~\Cref{lemma:WhatWeExpect},
\begin{align*}
    \OPT \quad &\leq \quad \sum_{T \subseteq [n]} \bigg((1+\delta)(1-p)^{n-t} p^t c \max\bigg\{\frac{\abs{T}}{\alpha(\abs{T})},\, \frac{t^*}{\alpha(t^*)}\bigg\}\bigg) \\
    &= \quad (1 + \delta) \sum_{t=0}^n \binom{n}{t} (1-p)^{n-t} p^t c \max\bigg\{\frac{\abs{T}}{\alpha(\abs{T})},\, \frac{t^*}{\alpha(t^*)}\bigg\}\bigg) \enspace .
\end{align*}

Appropriately scaling $\delta$, we see that $\beta_\delta(n)$ is precisely the ratio (minus a little bit) between when the communication problems are both $1$-instances and both $0$-instances, so since \textsc{PromiseDisjointness} requires $\Omega(z/n) = 2^{\Omega(c)}$ communication and \textsc{$(\calV, \alpha(\cdot))$-GapWelfare} requires $2^{\Omega(\sqrt{r})} = 2^{\Omega(c)}$ communication, we have by~\Cref{lemma:TwoProblems} that a $\beta(n, p, \delta)$-approximation for \VSM{} requires $2^{\Omega(c)}$ communication.
\end{proof}

\SeparationSASM*
\begin{proof}
Let $n = 3$, $p = 4/5$, $\delta = 0.01$, and $t^* = 1$. By~\Cref{thm:SubadditiveStronglyInapproximable}, we have that $t/\alpha(t) = 1 + \I[t = 3]/2 + o(1)$. Plugging everything in, we get a hardness of approximation of at least
\[
    \frac{(4/5)(3) + (1/5) (1 + o(1))}{(4/5)^3 (3/2 + o(1)) + (1 - (4/5)^3) (1 + o(1))} - 0.01 \quad \geq \quad 2.06 \enspace . \qedhere
\]
\end{proof}

\SeparationXOSSM*
\begin{proof}
Let $\beta(n) \coloneqq \beta(n, n/(n+1), 0.001)$. The theorem is verified by a direct computation of $\beta(n)$ for all $n \in [2, 150]$ (using~\Cref{prop:XOSStronglyInapproximable2} to compute $t/\alpha(t)$), and additionally verifying that $\beta(150) \geq e/(e - 1) + 0.001$. Sample code can be found in~\Cref{sec:code}.
\end{proof}

\section{Strong Inapproximability of Subadditive Valuations}
\label{sec:inapproximable}

In this section, we prove two strong inapproximability results for subadditive valuations, which are used to prove the hardness of approximation for \SASM{} and the separation in approximability between \SASM{} and \SA{}. We first restate the results and relevant definitions.

\Scarce*

\StronglyInapproximable*

\SubadditiveStronglyInapproximableProp*

\SubadditiveStronglyInapproximableThm*

The first result involves a standard reduction from promise disjointness using ``one-two'' subadditive valuations. The second result is far more involved, with a reduction from the multi-player generalization of the \EFS{} communication problem (the $2$-player variant having been introduced in~\cite{EzraFNTW19}) using modified set cover subadditive valuations (also introduced in~\cite{EzraFNTW19}).

Notice the qualitative differences between the two results: the first gives stronger convergence to $2$-inapproximability when $t$ is large, but suffers a constant gap when $t = O(1)$. This is suitable for our hardness of approximation results, where we deal with large numbers of bidders.

On the other hand, the second result gives a $(2-o(1))$-inapproximability which holds for all $t \in [2, n]$ when $n = 2^{o(\log(m))}$. In other words, the result is stronger in the sense that for \emph{any} number of bidders, the sub-problem on \emph{any} non-trivial subset of the bidders is hard. This is suitable for our separation results, where we need to account for small $n$ to show that welfare maximization over \SASM{} is strictly harder than over \SA{} for all $n \geq 3$. The property of ``every sub-problem is as hard as the main problem'' may also be of interest in other settings.

\subsection{Strong \texorpdfstring{$(2 - 2/(t+1))$}{(2 - 2/(t+1))}-Inapproximability}

\begin{proof}[Proof of~\Cref{prop:SubadditiveStronglyInapproximable}]
For $z = 2^{\Theta(m/n^2)}$, sample uniformly random partitions $\vb{A}_1, \dots,\vb{A}_z \in \Sigma$. For all distinct $i, j \in N$ and distinct $k, \ell \in [z]$, we have $Z_{i,j,k,\ell} \coloneqq \abs{A_{k,i} \cap A_{\ell,j}} \sim \Binom(m, 1/n^2)$. Then by a Chernoff bound, $\Pr[Z_{i,j,k,\ell} = 0] \leq 2^{-\Omega(m/n^2)}$. Union bounding over all at most $n^2 z^2 = 2^{\Theta(m/n^2)}$ choices of $i, j, k, \ell$, we have that w.h.p., $Z_{i,j,k,\ell} \ne 0$ for all distinct $i, j \in N$ and distinct $k, \ell \in [z]$ simultaneously. Therefore, there exist partitions $\vb{A}_1, \dots, \vb{A}_z \in \Sigma$ satisfying this ``Intersecting Property;'' fix any such partitions.

Consider a \textsc{PromiseDisjointness} instance with inputs $X_1, \dots, X_n \subseteq [z]$. For $v_1, \dots, v_n : 2^M \to \mathbb{R}_+$, define $v_i(S) \coloneqq \I[S \ne \emptyset] + \max_{\ell \in X_i} \I[S \supseteq A_{\ell,i}]$, and observe that $v_i$ is subadditive.

\begin{lemma} \label{lemma:KindaAdditive}
    If $X_1, \dots, X_n$ are pairwise disjoint, then for any allocation $\vb{A} \in \Sigma$ and $i, j \in N$, $v_i(A_i) + v_j(A_j) \leq 3$.
\end{lemma}
\begin{proof}
Suppose for contradiction that $v_i(A_i) = v_j(A_j) = 2$. Then there exist $k \in X_i, \ell \in X_j$ such that $A_i \supseteq X_{k,i}, A_j \supseteq X_{\ell,j}$, and $k \ne \ell$ because $X_i \cap X_j = \emptyset$. But by the Intersecting Property, we know that $X_{k,i} \cap X_{\ell,j} \ne \emptyset$, contradicting that $\vb{A}$ was a valid allocation.
\end{proof}

Suppose there exists $\ell \in \bigcap_{i \in N} X_i$. Then the allocation $\vb{A}_\ell$ yields $2n$ welfare.

Suppose $X_1, \dots, X_n$ are pairwise disjoint, and consider any $t$-scarce allocation. If no bidder has value $2$, then the welfare is at most $t$. If some bidder has value $2$, then by~\Cref{lemma:KindaAdditive}, no other bidder has value $2$, so the welfare is at most $t+1$.

Therefore, any algorithm which decides if there exists a $t$-scarce allocation with more than a $(t+1)/(2n)$ fraction of the optimal welfare is able to solve \textsc{PromiseDisjointness} on $[z]$, which requires $\Omega(z/n) = 2^{\Omega(m/n^2)}$ communication.
\end{proof}

\subsection{Strong \texorpdfstring{$(2 - \I[t = 1] - O(\log(n)/\log(m)))$}{(2 - O(log(n)/log(m))}-Inapproximability}

As mentioned previously, the second result uses a reduction from a custom communication problem, \EFS{}. We first formally state the problem.

For a set $S \in 2^M$, let $\closure{S}$ denote its complement, and for a list of sets $\vb{S} \in (2^M)^z$, let $\vb{\closure{S}}$ denote its set-wise complement.

\begin{definition}[Set Cover]
    A collection of sets $\calS \subseteq 2^M$ \emph{covers} a set $T \subseteq M$ if $\bigcup_{S \in \calS} S \supseteq T$. For a collection of sets $\calS \subseteq 2^M$ covering $M$, define $\varphi_\calS(T)$ to be the minimum number of sets from $\calS$ needed to cover $T$.
\end{definition}

\begin{definition}[$\lambda$-Sparse]
    For $\lambda \in \mathbb{Z}_+$, a collection of sets $\calS \subseteq 2^M$ is \emph{$\lambda$-sparse} if $\varphi_\calS(M) > \lambda$.
\end{definition}

\begin{definition}[\FS{}]
    For $m, n \in \mathbb{Z}_+$ such that $m = n \bmod{n^2}$, \FS{m,n} is an $n$ player communication problem where:
    \begin{itemize}[itemsep=0pt,topsep=0pt]
        \item \textbf{Input:} $X_1, \dots, X_n \in 2^M$.

        \item \textbf{Output:} $\textsc{Disj}(X_1, \dots, X_n)$.

        \item \textbf{Promise 1:} $\abs{X_i} = m/n$ for all $i \in [n]$.
        
        \item \textbf{Promise 2:} Either $\abs{X_1 \cap \dots \cap X_n} = m/n - n$, or $X_1, \dots, X_n$ are pairwise disjoint.
    \end{itemize}
\end{definition}

\begin{definition}[\EFS{}] \label{def:EFS}
    For $m, n, z \in \mathbb{Z}_+$ such that $m = n \bmod{n^2}$, \EFS{m,n,z} is an $n$ player communication problem where:
    \begin{itemize}[itemsep=0pt,topsep=0pt]
        \item \textbf{Input:} $\vb{X}_1, \dots, \vb{X}_n \in (2^M)^z$.

        \item \textbf{Output:} $\bigvee_{\ell \in [z]} \textsc{Disj}(X_{1,\ell}, \dots, X_{n,\ell})$.

        \item \textbf{Promise 1:} $\abs{X_{i,\ell}} = m/n$ for all $i \in [n], \ell \in [z]$.

        \item \textbf{Promise 2:} For each $\ell \in [z]$, either $\abs{X_{1,\ell} \cap \dots \cap X_{n,\ell}} = m/n - n$, or $X_{1,\ell}, \dots, X_{n,\ell}$ are pairwise disjoint.

        \item \textbf{Promise 3:} $\vb{\closure{X}}_i$ is $\log(m^{1/3})/\log(2n)$-sparse for all $i \in [n]$.

        \item \textbf{Promise 4:} For all $i \in [n]$ and $S \subseteq M$ of size $n\log(m)$, there exists $\ell \in [z]$ such that $S \subseteq \closure{X}_{i,\ell}$.
    \end{itemize}
\end{definition}

Observe that \EFS{m,n,z} is the OR of $z$ \FS{m,n} instances, with the inclusion of an additional promise on each player's input. The promises in \EFS{} are necessary to construct a hard subadditive instance from the communication problem. The following is a lower bound on the communication complexity of \EFS{}, which we prove in~\Cref{sec:exists-far-sets}.

\begin{theorem}[name=,restate=ExistsFarSets] \label{thm:EFS}
    For $z = 2^{\Theta(\sqrt{m})}$, the randomized communication complexity of \emph{\EFS{m,n,z}} is $2^{\Omega(\sqrt{m}-n^2\log(m))}$.
\end{theorem}

Our lower bound construction is a natural generalization of the $2$-bidder~\cite{EzraFNTW19} construction to $n$ bidders, and so follows a similar proof sketch. Let $\lambda \coloneqq \log(m^{1/3})/\log(2n)$ be an even integer. For $\vb{X} \in (2^M)^z$, the idea is to design a subadditive function $f_{\vb{X}}$ such that
\begin{enumerate}[topsep=4pt]
    \item $f_{\vb{X}}(X_\ell) \approx \lambda$ for all $\ell \in [z]$.
    
    \item $f_{\vb{X}}(S) + f_{\vb{X}}(\bar{S}) = \lambda$ for all $S \subseteq M$.
\end{enumerate}

Then for an \EFS{m,n,z} input $\vb{X}_1, \dots, \vb{X}_n$, we have the hard instance $f_{\vb{X}_1}, \dots, f_{\vb{X}_n}$. If we are in a $1$-instance, then by $(1)$, we can achieve welfare $\approx n\lambda$ by allocating $X_{1,\ell}, \dots, X_{n,\ell}$ for the $\ell \in [z]$ where the sets are disjoint. If we are in a $0$-instance, suppose for contradiction that there exists a $t$-scarce allocation achieving much more than $\lambda t/2$ welfare. Then some bidder must receive a set $S$ with at least $\lambda/2$ value, and everyone else receives at most $\closure{S}$. But since in a $0$-instance, the inputs $\vb{X}_1, \dots, \vb{X}_n$ are ``almost equal,'' we have by $(2)$ that everyone else gets value at most $\approx \lambda/2$ value, and hence the optimal $t$-scarce allocation cannot exceed $\approx \lambda t/2$ welfare.

We now formally define the function family and make these notions rigorous.

\begin{definition}
    For $\vb{X} \in (2^M)^z$, define $f_{\vb{X}} : 2^M \to \mathbb{R}_+$ by
    \[
        f_{\vb{X}}(S) \quad \coloneqq \quad \begin{cases}
            \varphi_{\vb{\closure{X}}}(S) & \varphi_{\vb{\closure{X}}}(S) < \lambda/2 \\
            \lambda - \varphi_{\vb{\closure{X}}}(\closure{S}) & \varphi_{\vb{\closure{X}}}(\closure{S}) < \lambda/2 \\
            \lambda/2 & o.w.
        \end{cases} \enspace .
    \]
\end{definition}

\begin{proposition}[\cite{EzraFNTW19}] \label{prop:SubadditiveFunction}
    If $\vb{\closure{X}} \in (2^M)^z$ is $\lambda$-sparse,
    \begin{enumerate}[itemsep=0pt,topsep=0pt]
        \item $f_{\vb{X}}$ is well-defined and subadditive.
        \item $f_{\vb{X}}(X_\ell) = \lambda - 1$ for all $\ell \in [z]$.
        \item $f_{\vb{X}}(S) > \lambda/2$ implies $f_{\vb{X}}(S) = \lambda - \varphi_{\vb{\closure{X}}}(\closure{S})$.
    \end{enumerate}
\end{proposition}

We are now ready for the proof of the main result.

\begin{proof}[Proof of~\Cref{thm:SubadditiveStronglyInapproximable}]
Let $z = 2^{\Theta(\sqrt{m})}$, and consider an \EFS{m,n,z} instance $\vb{X}_1, \dots, \vb{X}_n \in (2^M)^z$. For valuations $v_1, \dots, v_n : 2^M \to \mathbb{R}_+$, define $v_i \coloneqq f_{\vb{X}_i}$, and observe that $v_i$ is subadditive by Promise $(3)$ of~\Cref{def:EFS} and Property $(1)$ of~\Cref{prop:SubadditiveFunction}.

\begin{lemma} \label{lemma:AlmostPerfectlyAdditive}
    If $\vb{X}_1, \dots, \vb{X}_n$ is a $0$-instance, then for any allocation $\vb{A} \in \Sigma$ and $i, j \in N$, $v_i(A_i) + v_j(A_j) \leq \lambda + 2$.
\end{lemma}
\begin{proof}
If $v_i(A_i) \leq \lambda/2 + 1$ and $v_j(A_j) \leq \lambda/2 + 1$, the lemma follows immediately, so suppose WLOG that $v_i(A_i) \geq \lambda/2 + 2$. Then by Property $(3)$ of \Cref{prop:SubadditiveFunction}, $v_i(A_i) = \lambda - \varphi_{\vb{\closure{X}}_i}(\closure{A}_i)$, implying $\varphi_{\vb{\closure{X}}_i}(\closure{A}_i) \leq \lambda - v_i(A_i)$. Let the indices of the $\lambda - v_i(A_i)$ sets from $\vb{\closure{X}}_i$ covering $\closure{A}_i$ be $L \subseteq [z]$.

By Promises $(1)$ and $(2)$ of \EFS{}, we have for all $\ell \in L$ that $\closure{X}_{j,\ell}$ is missing at most $n$ items from $\closure{X}_{i,\ell}$. Therefore, the collection $\{\closure{X}_{j,\ell} : \ell \in L\}$ covers all but at most $n(\lambda - v_i(A_i)) < n\log(m)$ items $S$ of $\closure{A}_i$. By Promise $(4)$ of \EFS{}, there exists a single set in $\vb{\closure{X}}_j$ which covers $S$, so $\varphi_{\vb{\closure{X}}_j}(\closure{A}_i) \leq \varphi_{\vb{\closure{X}}_i}(\closure{A}_i) + 1 < \lambda/2$. Thus,
\[
    v_j(A_j) \quad \leq \quad v_j(\closure{A}_i) \quad = \quad \varphi_{\vb{\closure{X}}_j}(\closure{A}_i) \quad \leq \quad \varphi_{\vb{\closure{X}}_i}(\closure{A}_i) + 1 \quad = \quad \lambda - v_i(A_i) + 1 \enspace ,
\]
which implies $v_i(A_i) + v_j(A_j) \leq \lambda + 1$.
\end{proof}

Suppose $\vb{X}_1, \dots, \vb{X}_n$ is a $1$-instance. Then there exists $\ell \in [z]$ such that $X_{1,\ell}, \dots, X_{n,\ell}$ are pairwise disjoint. By Property $(2)$ of \Cref{prop:SubadditiveFunction}, this allocation yields welfare $n(\lambda-1)$.

Suppose $\vb{X}_1, \dots, \vb{X}_n$ is a $0$-instance, and consider any $t$-scarce allocation. If no bidder has value more than $\lambda/2+1$, then the welfare is at most $t(\lambda/2+1)$. If some bidder has value $x \geq \lambda/2+1$, then by~\Cref{lemma:AlmostPerfectlyAdditive}, no other bidder has value more than $\lambda + 2 - x$, so the welfare is at most $\lambda + 2 + (t-2)(\lambda + 2 - x)$. Some casework bounds the welfare at $(t + \I[t = 1])(\lambda/2+1)$.

Thus, an algorithm which decides if there exists a $t$-scarce allocation with more than a $(t + \I[t = 1])(\lambda/2+1)/(n(\lambda-1)) = (t/n)/(2 - \I[t = 1] - \Theta(\lambda))$ fraction of the optimal welfare is able to solve \EFS{m,n,z}, which requires $2^{\Omega(\sqrt{m}-n^2\log(m))}$ communication by~\Cref{thm:EFS}.
\end{proof}

\section{Multi-Player \texorpdfstring{$\exists$}{Exists}FarSets}
\label{sec:exists-far-sets}

In this section, we prove a randomized communication lower bound for \EFS{} (\Cref{def:EFS}).

\ExistsFarSets*

As previously noted, observe that \EFS{m,n,z} is essentially the OR of $z$ \FS{m,n} instances. In general, solving the OR of $z$ problems does not require $z$ times the communication. For example, not-equality is the OR of $z$ XOR problems, but requires only $\Theta(\log(z))$ randomized communication. However, solving the OR of $z$ problems \emph{does} require $z$ times the \emph{information} (for an appropriately defined metric of information). This is useful because information lower bounds communication.

Using this principle,~\cite{EzraFNTW19} proves that the \emph{internal information complexity} of \FS{m,2} is non-negligible, then uses standard machinery from information theory to show that the internal information complexity (which lower bounds the randomized communication complexity) of \EFS{m,2,z} is $z$ times that of \FS{m,2}.

The main barrier in generalizing this approach to $n$ players is that because of secure multiparty computation, any protocol with at least $3$ players can be made to have $0$ internal information complexity. We overcome this by using \emph{external conditional information complexity} (referred to later as conditional information complexity), which works the same way morally, but requires additional care in the technical details.

\subsection{Information Theory Preliminaries}

We first overview standard information theory results and notation, drawing primarily from~\cite{BravermanO15} (but which would be typical in many information theory papers). Note that for the most part we will not be making use of the precise definitions of entropy, mutual information, and KL divergence, only the listed facts. For this section, let all random variables take values in a universe $\Omega$ with finite cardinality, and let $\log$ denote the binary logarithm.

\begin{definition}[Entropy]
    For random variables $(X, Y) \sim \mu$ with marginals $\mu_x, \mu_y$, the (Shannon) \emph{entropy} of $X$ is defined
    \[
        \H(X) \quad \coloneqq \quad \E[\log(1/\mu_x(X))] \enspace .
    \]
    The \emph{conditional entropy} of $X$ given $Y$ is defined
    \[
        \H(X \given Y) \quad \coloneqq \quad \E\bigg[\frac{\log(1/\mu(X, Y))}{\log(1/\mu_y(Y))}\bigg] \enspace .
    \]
\end{definition}

\begin{definition}[Mutual Information]
    For random variables $X, Y, Z$, the (conditional) \emph{mutual information} between $X$ and $Y$ given $Z$ is defined
    \[
        \MI(X \mutual Y \given Z) \quad \coloneqq \quad \H(X \given Z) - \H(X \given Y, Z) \quad = \quad \H(Y \given Z) - \H(Y \given X, Z) \enspace .
    \]
\end{definition}

\begin{fact}[Chain Rule]
    For random variables $X_1, \dots, X_n, Y, Z$,
    \[
        \MI(\vb{X} \mutual Y \given Z) \quad = \quad \sum_{i=1}^n \MI(X_i \mutual Y \given X_1, \dots, X_{i-1}, Z) \enspace .
    \]
\end{fact}

\begin{fact} \label{fact:MutualInformationFact1}
    For random variables $X, Y, Z, W$ where $X, Z$ are independent given $W$,
    \[
        \MI(X \mutual Y, Z \given W) \quad = \quad \MI(X \mutual Y \given Z, W) \enspace .
    \]
\end{fact}

\begin{fact} \label{fact:MutualInformationFact2}
    For random variables $X, Y, Z, W$ where $X, Z$ are independent given $W$,
    \[
        \MI(X \mutual Y \given W) \quad \leq \quad \MI(X \mutual Y \given Z, W) \enspace .
    \]
\end{fact}

\begin{definition}[KL Divergence]
    For random variables $X \sim \mu_x, Y \sim \mu_y$, the \emph{KL divergence} (also called \emph{relative entropy}) of $X$ from $Y$ is defined
    \[
        \KL(X \relative Y) \quad \coloneqq \quad \E\bigg[\log\bigg(\frac{\mu_x(X)}{\mu_y(X)}\bigg)\bigg] \enspace .
    \]
\end{definition}

\begin{fact} \label{fact:MutualInformationToKLDivergence}
    For random variables $(X, Y, Z) \sim \mu$, let $X_z$ denote the random variable $X$ given $Z = z$, and $X_{yz}$ denote the random variable $X$ given $Y = y$, $Z = z$. Then
    \[
        \MI(X \mutual Y \given Z) \quad = \quad \E_{(A, B, C) \sim \mu}\Big[\KL(X_{BC} \relative X_{C})\Big]
    \]
\end{fact}

\begin{definition}[Total Variation]
    For random variables $X \sim \mu_x, Y \sim \mu_x$, the \emph{total variation distance} between $X$ and $Y$ is defined
    \[
        \TV(X, Y) \;\;\, \coloneqq \;\;\, \max_{A \subseteq \Omega} \Big(\mu_x(X) - \mu_y(Y)\Big) \;\;\, = \;\;\, \sum_{\mathclap{\substack{\omega \in \Omega \\ \mu_x(\omega) \geq \mu_y(\omega)}}} \Big(\mu_x(\omega) - \mu_y(\omega)\Big) \;\;\, = \;\;\, \frac{1}{2}\sum_{\omega \in \Omega} \abs{\mu_x(\omega) - \mu_y(\omega)} \enspace .
    \]
\end{definition}

\begin{fact}[Pinsker's Inequality]
    For random variables $X, Y$,
    \[
        \TV(X, Y) \quad \leq \quad \sqrt{\KL(X \relative Y)/(2\log e)} \enspace .
    \]
\end{fact}

\begin{definition}[Conditional Information Complexity] \label{def:ConditionalInformationComplexity}
    Let $\calQ$ be an $n$-player communication problem where each player receives an input in $\{0, 1\}^m$. Let $\delta \in (0, 1)$, and let $\calP_\delta$ be the set of all randomized protocols which solve $\calQ$ with worst-case error $\delta$. For any randomized protocol $\Pi \in \calP_\delta$ and input $\vb{x} \in \{0, 1\}^{n \times m}$, let $\Pi(\vb{x})$ be the random variable of the transcript of $\Pi$ on input $\vb{x}$. Let $\mu$ be a distribution on $\{0, 1\}^{n \times m} \times \calD$ for some domain $\calD$. The \emph{conditional information complexity} of $\Pi$ with respect to $\mu$ is defined
    \[
        \CIC_\mu(\Pi) \quad \coloneqq \quad \MI_{(\vb{X}, D) \sim \mu}(\Pi(\vb{X}) \mutual \vb{X} \given D) \enspace .
    \]
    The \emph{conditional information complexity} of $\calQ$ with respect to $\mu,\delta$ is defined
    \[
        \CIC_{\mu,\delta}(\calQ) \quad \coloneqq \quad \inf_{\Pi \in \calP_\delta} \CIC_\mu(\Pi) \enspace .
    \]
\end{definition}

\begin{fact} \label{fact:InformationLowerBoundsCommunication}
    Let $\calQ$ be a communication problem, let $\delta \in (0, 1)$, and let $\CC_\delta(\calQ)$ denote its randomized communication complexity with worst-case error $\delta$. Then for any $\mu$,
    \[
        \CIC_{\mu,\delta}(\calQ) \quad \leq \quad \CC_\delta(\calQ) \enspace .
    \]
\end{fact}

\subsection{\texorpdfstring{$\exists$}{Exists}FarSets to FarSets: a Direct Sum Theorem}

As a reminder, the previous fact is why information complexity is relevant: it lower bounds the communication complexity. The next theorem (preceded by some notational exposition) justifies looking at information complexity instead of communication complexity: information complexity satisfies a ``direct sum'' property, where solving the OR of $z$ problem instances requires $z$ times the information as solving a single problem instance. This allows us to reduce the problem of lower bounding the information complexity of \EFS{} to lower bounding the information complexity of \FS{}. The following notation and proof approach follows that of~\cite{EzraFNTW19,BravermanO15,Bar-YossefJKS02}.

\begin{definition}[Promise Problem]
    Let $f : \{0, 1\}^{n \times m} \to \{0, 1\}$ and let $\calX \subseteq \{0, 1\}^{n \times m}$. The \emph{communication problem solving $f$ under promise $\calX$} refers to the communication problem where the players must output $f(\vb{x})$ when $\vb{x} \in \calX$ and may provide arbitrary output otherwise.
\end{definition}

\begin{definition}[$(z, \varepsilon)$-Safe]
    Let $\calX \subseteq \{0, 1\}^{n \times m}$ and $\mu$ be a distribution on $\calX$, and let $\calX^* \subseteq \{0, 1\}^{n \times z \times m}$. We say that $\mu, \calX$ are \emph{$(z, \varepsilon)$-safe} with respect to $\calX^*$ if for any $k \in [z]$, $\vb{X}_k \in \calX$, and $\vb{X}_1, \dots, \vb{X}_z$, where $\vb{X}_\ell \sim \mu$ independently for $\ell \ne k$,
    \[
        \Pr\Big[\Big((X_{1,1}, \dots, X_{z,1}), \dots, (X_{1,n}, \dots, X_{z,n})\Big) \not\in \calX^*\Big] \quad \leq \quad \varepsilon \enspace .
    \]
\end{definition}

\begin{theorem}[Direct Sum] \label{thm:DirectSum}
    Let $f : \{0, 1\}^{n \times m} \to \{0, 1\}$, $\calX \subseteq \{0, 1\}^{n \times m}$, and $\calX^* \subseteq \{0, 1\}^{n \times z \times m}$. Let $\mu$ be a distribution on $\{0, 1\}^{n \times m} \times \calD$ satisfying the following properties:
    \begin{enumerate}[topsep=5pt,itemsep=0pt]
        \item $\mu, \calX$ are $(z, \varepsilon)$-safe with respect to $\calX^*$.
        
        \item For $(\vb{X}, D) \sim \mu$, $\vb{X} \in \calX$ and $f(\vb{X}) = 0$ a.s.

        \item For $(\vb{X}, D) \sim \mu$, $X_1, \dots, X_n$ are independent given $D$.
    \end{enumerate}
    Let $\mu^*$ be the distribution that independently samples $(\vb{X}_1, D_1), \dots, (\vb{X}_z, D_z) \sim \mu$ and outputs $(((X_{1,1}, \dots, X_{z,1}), \dots, (X_{1,n}, \dots, X_{z,n})), (D_1, \dots, D_z))$. Let $f^* : \{0, 1\}^{n \times z \times m} \to \{0, 1\}$ be defined
    \[
        f^*\Big((X_{1,1}, \dots, X_{1,z}), \dots, (X_{n,1}, \dots, X_{n,z})\Big) \quad \coloneqq \quad \bigvee_{\ell \in [z]} f(X_{1,\ell}, \dots, X_{n,\ell}) \enspace .
    \]
    Let $\Pi^*$ be a protocol solving $f^*$ with worst-case error $\delta$ such that $\CIC_{\mu^z}(\Pi^*) \leq zc$. Then there exists a protocol $\Pi$ solving $f$ with worst-case error $\delta + \varepsilon$ such that $\CIC_\mu(\Pi) \leq c$.
\end{theorem}
\begin{proof}
For $\vec{\vb{x}} \in \{0, 1\}^{n \times z \times m}$, define $\vec{x}_\ell \in \{0, 1\}^{n \times m}$ to be a slice along the $[z]$ dimension.

Given input $\vec{x} \in \calX$, $\Pi$ is the protocol which does the following.

The players agree on a uniformly sampled $k \in [z]$ and $\vb{D}_{-k} \in \calD^{z-1}$ sampled from $\mu^{z-1}$ using public randomness. Then for each $\ell \in [z] \setminus \{k\}$, each player $i$ samples $X_{i,\ell} \in \{0, 1\}^m$ from $\mu$ conditioned on $D_\ell$, which is possible because for $(\vb{X}, D) \sim \mu$, $X_1, \dots, X_n$ are independent given $D$. Each player then sets $\vb{X}_i \coloneqq (X_{i,1}, \dots, X_{i,k-1}, x_i, X_{i,k+1}, \dots, X_{i,z})$.

The players then run $\Pi^*$ on $\vec{\vb{X}}$. Since $\mu$ is distributed on $0$ inputs to $f$, $f^*(\vec{\vb{X}}) = 1$ if and only if $f(\vec{x}) = 1$. Therefore, $\Pi$ only errs if $\Pi^*$ errs, or if $\vec{\vb{X}} \not\in \calX^*$. The former occurs w.p. $\leq \delta$, and the latter occurs w.p. $\leq \varepsilon$ by $(z, \varepsilon)$-safety of $\mu, \calX$, so $\Pi$ errs w.p. at most $\delta + \varepsilon$.

Now, let us analyze $\CIC_\mu(\Pi)$. Observe that the transcript of $\Pi$ is comprised of $k$ and $\vb{D}_{-k}$ (generated via public randomness), and the transcript of $\Pi^*$. Hence,
\begin{align*}
    \CIC_\mu(\Pi) \quad &= \quad \MI_{(\vec{X}, D) \sim \mu}(\vb{X} \mutual \Pi(\vec{X}) \given D) \\
    &= \quad \MI_{k \sim [z], (\vec{\vb{X}}, \vb{D}) \sim \mu^*}(\vec{X}_k \mutual \Pi^*(\vec{\vb{X}}), k, \vb{D}_{-k} \given D_k) \enspace .
\end{align*}

Since $k$ is drawn independently and $\mu^*$ consists of independent draws from $\mu$, $\vec{X}_k$ is independent from $k$ and independent from $\vec{X}_\ell, D_\ell$ for all $\ell \ne k$. Thus, by \Cref{fact:MutualInformationFact1}, \Cref{fact:MutualInformationFact2}, writing out the dependence on $k$, and the chain rule, we have
\begin{align*}
    \CIC_\mu(\Pi) \quad &= \quad \MI_{k \sim [z], (\vec{\vb{X}}, \vb{D}) \sim \mu^*}(\vec{X}_k \mutual \Pi^*(\vec{\vb{X}}) \given k, \vb{D}) \\
    &\leq \quad \MI_{k \sim [z], (\vec{\vb{X}}, \vb{D}) \sim \mu^*}(\vec{X}_k \mutual \Pi^*(\vec{\vb{X}}) \given k, \vec{X}_1, \dots, \vec{X}_{k-1}, \vb{D}) \\
    &= \quad \frac{1}{z} \sum_{\ell=1}^z \MI_{(\vec{\vb{X}}, \vb{D}) \sim \mu^*}(\vec{X}_\ell \mutual \Pi^*(\vec{\vb{X}}) \given \vec{X}_1, \dots, \vec{X}_{\ell-1}, \vb{D}) \\
    &= \quad \frac{1}{z} \MI_{(\vec{\vb{X}}, \vb{D}) \sim \mu^*}(\vec{\vb{X}} \mutual \Pi^*(\vec{\vb{X}}) \given \vb{D}) \\
    &= \quad \frac{1}{z} \CIC_{\mu^*}(\Pi^*) \quad \leq \quad c \enspace . \qedhere
\end{align*}
\end{proof}

\begin{corollary} \label{cor:DirectSum}
    $\CC_\delta(f^*) \geq \CIC_{\mu^*,\delta}(f^*) \geq z \cdot \CIC_{\mu,\delta+\varepsilon}(f)$.
\end{corollary}
\begin{proof}
Follows from \Cref{fact:InformationLowerBoundsCommunication} and \Cref{def:ConditionalInformationComplexity}.
\end{proof}

\subsection{Information Complexity of FarSets}

The final piece of the communication lower bound is a lower bound on the information complexity of \FS{} with respect to a distribution satisfying the requirements of \Cref{thm:DirectSum}. As before, the proof approach follows that of~\cite{EzraFNTW19}, but differs in the technical details due to using external instead of internal information complexity.

The proof sketch looks like a hybrid argument from cryptography:
\begin{enumerate}[topsep=4pt]
    \item Suppose that the information complexity of a protocol is small on a distribution of $0$-instances.

    \item There is just enough freedom in the promise of what a $0$-instance looks like, so that we can \emph{chain} together many $0$-instances (call these instances \emph{links}), with each adjacent instance in the chain differing only by one player's input.

    \item By a counting argument, we show that if the information complexity is low, there exists some chain where the protocol reveals little information for each of its links. Fix any such chain.

    \item Using standard information theory, little information revealed means the distribution of the transcripts for every link in the chain must be close to each other.

    \item By what is essentially the randomized analogue of a deterministic rectangle argument, we argue that even if we mix inputs from different links in the chain, the resulting transcript must still be distributionally close. This step is the main technical component of the proof.

    \item If we have constructed the chain such that a $1$-instance can be made by mixing inputs from different links, then there exists a $0$-instance and a $1$-instance such that their transcripts (and hence the output of the protocol) are distributionally close, which implies the worst-case error of the protocol must be high.

    \item We conclude that any protocol with low worst-case error must have high information complexity (relatively speaking).
\end{enumerate}

Let $m \equiv n \bmod{n^2}$ be sufficiently large, and let $k \coloneqq m/n \equiv 1 \bmod{n}$.

Let $\mu$ be a distribution on $(2^M)^n \times 2^M$, where for $(\vb{X}, D) \sim \mu$, $D$ is a uniformly random subset of $M$ of size $k-n$, and each $X_i$ is (independently) a uniformly random superset of $D$ of size $k$.

Fix a protocol $\Pi$, and in addition to $\Pi(\vb{X})$ denoting the random variable of the transcript on input $\vb{X}$, let $\Pi(D)$ for $\abs{D} = k-n$ denote the random variable of the transcript on an input sampled from $\mu$ conditioned on $D$.

Let $S_1, \dots, S_{k(n-1)+1}$ be sets defined $S_j \coloneqq [j, j+k-1]$. For any permutation $\sigma : M \to M$ on items, we abuse notation to also let $\sigma$ act on sets and tuples of sets item-wise.

\begin{definition}[Link]
    $\vb{L} \in (2^M)^n$ is a \emph{link} if there exists a permutation $\sigma : M \to M$, $a \in [k-1]$, and $b \in [n]$ such that
    \[
        \sigma(L_i) \quad = \quad \begin{cases}
            S_{an+i} & i < b \\
            S_{(a-1)n+i} & i \geq b
        \end{cases} \enspace .
    \]
\end{definition}

\begin{definition}[Chain] \label{def:Chain}
    A \emph{chain} is a sequence of links $\vb{L}_{n+1}, \dots, \vb{L}_{(k-1)(n-1)+n+1}$ such that there exists a permutation $\sigma : M \to M$ where for all $a \in [k-1], b \in [n]$ resulting in valid indices,
    \[
        \sigma(L_{an+b,i}) \quad = \quad \begin{cases}
            S_{an+i} & i < b \\
            S_{(a-1)n+i} & i \geq b
        \end{cases} \enspace .
    \]
\end{definition}

\begin{example}
    For $M = \{1, 2, 3, 4, 5, 6, 7, 8, 9, 0, e, \pi\}$ and $n = 3$, the links associated with the identity permutation are
    \begin{itemize}[itemsep=0pt,topsep=0pt]
        \item $a, b = 1, 1$: $(\{1, 2, 3, 4\}, \{2, 3, 4, 5\}, \{3, 4, 5, 6\})$

        \item $a, b = 1, 2$: $(\{4, 5, 6, 7\}, \{2, 3, 4, 5\}, \{3, 4, 5, 6\})$

        \item $a, b = 1, 3$: $(\{4, 5, 6, 7\}, \{5, 6, 7, 8\}, \{3, 4, 5, 6\})$

        \item $a, b = 2, 1$: $(\{4, 5, 6, 7\}, \{5, 6, 7, 8\}, \{6, 7, 8, 9\})$

        \item $a, b = 2, 2$: $(\{7, 8, 9, 0\}, \{5, 6, 7, 8\}, \{6, 7, 8, 9\})$

        \item $a, b = 2, 3$: $(\{7, 8, 9, 0\}, \{8, 9, 0, e\}, \{6, 7, 8, 9\})$

        \item $a, b = 3, 1$: $(\{7, 8, 9, 0\}, \{8, 9, 0, e\}, \{9, 0, e, \pi\})$,
    \end{itemize}
    and this sequence of links forms the chain associated with the identity permutation.
\end{example}

\begin{remark}
    Observe that there is a bijection between chains and permutations, so we may refer to a chain by its permutation.
\end{remark}

\begin{definition}[Broken]
    A link $\vb{L}$ is \emph{$\gamma$-broken} if $\KL(\Pi(\vb{L}) \relative \Pi(D)) \geq \gamma$ for some $D \subseteq \bigcap_{i \in [n]} L_i$ of size $k-n$. A chain is \emph{$\gamma$-broken} if at least one of its links is $\gamma$-broken.
\end{definition}

\begin{lemma} \label{lemma:LowInformationMeansFewBrokenLinks}
    If $\CIC_\mu(\Pi) \leq \alpha\gamma/m^{n^2}$, then at most an $\alpha$ fraction of links are $\gamma$-broken.
\end{lemma}
\begin{proof}
Let $\ell$ be the number of links, and observe that $\ell \geq \binom{m}{k-n+1}$, because for every set $E$ of $k-n+1$ items, there exists at least one link $\vb{L}$ such that $L_1 \cap \dots \cap L_n = E$, and clearly no link can satisfy this for two distinct $E, E'$.

Then by \Cref{fact:MutualInformationToKLDivergence} and the previous fact,
\begin{align*}
    \CIC_\mu(\Pi) \quad &= \quad \MI_{(\vb{X}, D) \sim \mu}(\Pi(\vb{X}) \mutual \vb{X} \given D) \\
    &= \quad \E_{(\vb{X}, D) \sim \mu}[\KL(\Pi(\vb{X}) \relative \Pi(D))] \\
    &= \quad \frac{1}{\binom{m}{k-n} \binom{m-k+n}{n}^n} \sum_{\substack{D \subseteq M \\ \abs{D} = k-n}} \sum_{\substack{X_1, \dots, X_n \supseteq D \\ \abs{X_1} = \dots = \abs{X_n} = k}} \KL(\Pi(\vb{X}) \relative \Pi(D)) \\
    &\geq \quad \frac{1}{\ell m^{n^2}} \sum_{\substack{D \subseteq M \\ \abs{D} = k-n}} \sum_{\substack{X_1, \dots, X_n \supseteq D \\ \abs{X_1} = \dots = \abs{X_n} = k}} \KL(\Pi(\vb{X}) \relative \Pi(D)) \\
    &\geq \quad \frac{1}{\ell m^{n^2}} \sum_{\substack{D \subseteq M \\ \abs{D} = k-n}} \sum_{\substack{L_1, \dots, L_n \supseteq D \\ \text{$\vb{L}$ is a link}}} \KL(\Pi(\vb{L}) \relative \Pi(D)) \\
    &= \quad \frac{1}{\ell m^{n^2}} \sum_{\text{$\vb{L}$ is a link}} \sum_{\substack{D \subseteq L_1 \cap \dots \cap L_n \\ \abs{D} = k-n}} \KL(\Pi(\vb{L}) \relative \Pi(D)) \\
    &\geq \quad \frac{1}{\ell m^{n^2}} \sum_{\text{$\vb{L}$ is a link}} \max_{\substack{D \subseteq L_1 \cap \dots \cap L_n \\ \abs{D} = k-n}} \KL(\Pi(\vb{L}) \relative \Pi(D)) \enspace .
\end{align*}

Observe that this last quantity is just $1/m^{n^2}$ times the ``average brokenness'' of a link. Therefore, if more than an $\alpha$ fraction of links are $\gamma$-broken, $\CIC_\mu(\Pi) \geq \alpha\gamma/m^{n^2}$.
\end{proof}

\begin{lemma} \label{lemma:FewBrokenLinksMeansUnbrokenChain}
    If $\CIC_\mu(\Pi) \leq \gamma/m^{n^2+1}$, then there exists a $\gamma$-unbroken chain.
\end{lemma}
\begin{proof}
Let $\ell$ be the number of links and $c$ be the number of chains. Since each chain consists of $(k-1)(n-1)+1$ links, there are $((k-1)(n-1)+1)c$ link-chain pairs where the link is in the chain. By symmetry, each link is part of the same number of chains, so each link is part of $((k-1)(n-1)+1)c/\ell$ chains. Thus, to $\gamma$-break all $c$ chains, at least $\ell/((k-1)(n-1)+1) \geq \ell/m$ links must be $\gamma$-broken, so if fewer than a $1/m$ fraction of links are $\gamma$-broken, there exists a $\gamma$-unbroken chain. Applying \Cref{lemma:LowInformationMeansFewBrokenLinks} with $\alpha = 1/m$ completes the proof.
\end{proof}

\begin{lemma} \label{lemma:UnbrokenChainMeansCloseLinks}
    If $\vb{L}_{n+1}, \dots, \vb{L}_{(k-1)(n-1)+n+1}$ is a $\gamma$-unbroken chain, then $\TV(\Pi(\vb{L}_j), \Pi(\vb{L}_{j+1})) \leq \sqrt{2\gamma}$ for all $j \in [n+1, (k-1)(n-1)+n)]$.
\end{lemma}
\begin{proof}
WLOG, let the chain be specified by the identity permutation (otherwise permute the items). Fix any $j \in [n+1, (k-1)(n-1)+n]$, and let

\[
    D \quad \coloneqq \quad \bigg(\bigcap_{i \in [n]} L_{j,i}\bigg) \cap \bigg(\bigcap_{i \in [n]} L_{j+1,i}\bigg) \quad = \quad \bigcap_{\ell=j-n}^j S_\ell \enspace .
\]

Observe that $\abs{D} = k-n$. By definition of $\gamma$-unbroken, $KL(\Pi(\vb{L}_j) \relative \Pi(D)), \KL(\Pi(\vb{L}_{j+1}) \relative \Pi(D)) < \gamma$. Therefore, by triangle inequality and Pinsker's Inequality,
\[
    \TV(\Pi(\vb{L}_j), \Pi(\vb{L}_{j+1})) \quad \leq \quad \TV(\Pi(\vb{L}_j), \Pi(D)) + \TV(\Pi(\vb{L}_{j+1}), \Pi(D)) \quad \leq \quad \sqrt{2\gamma} \enspace . \qedhere
\]
\end{proof}

The following result is the randomized protocol equivalent of a rectangle argument for deterministic protocols, and the proof follows by observing that at each point of the protocol, the next addition to the transcript depends only on the current transcript and the input of the player next to act. See a full proof in, e.g.,~\cite{BravermanO15}.

\begin{lemma} \label{lemma:RandomizedRectangle}
    There exist functions $\rho_1, \dots, \rho_n$ such that for any input $\vb{X}$ and transcript $\pi$,
    \[
        \Pr[\Pi(\vb{X}) = \pi] \quad = \quad \prod_{i \in [n]} \rho_i(X_i, \pi) \enspace .
    \]
\end{lemma}

The following theorem is the main technical result, and corresponds to $(5)$ in the proof sketch.

\begin{theorem}
    If $\CIC_\mu(\Pi) \leq 1/m^{n^2+9}$, then there exists a chain $\vb{L}_{n+1}, \dots \vb{L}_{(k-1)(n-1)+n+1}$ such that for $\vb{X}^{(0)} \coloneqq \vb{L}_{n+1}$ and $\vb{X}^{(1)} \coloneqq (L_{(k-1)(i-1)+n+1,i} : i \in [n])$,
    \[
        \TV(\Pi(\vb{X}_0), \Pi(\vb{X}_1)) \quad = \quad O(1/m) \enspace .
    \]
\end{theorem}
\begin{proof}
By \Cref{lemma:FewBrokenLinksMeansUnbrokenChain}, there exists a $1/m^8$-unbroken chain $\vb{L}_{n+1}, \dots, \vb{L}_{(k-1)(n-1)+n+1}$. WLOG, let this chain be specified by the identity permutation. Observe that in the chain, each player $i$ only ever holds as input the sets $S_j$ such that $j \equiv i \bmod{n}$.

Then for all $j \in [1, k(n-1)+1]$, let $i(j) \in [n]$ be such that $j \equiv i \bmod{n}$, and let $p_j(\pi) \coloneqq \rho_{i(j)}(S_j, \pi)$. Additionally, let
\begin{align*}
    s(\pi) \quad &\coloneqq \quad \sum_{j=0}^{k(n-1)-n} \bigg\vert\prod_{i \in [n]} p_{j+i}(\pi) - \prod_{i \in [n]} p_{j+i+1}(\pi)\bigg\vert \\
    &= \quad \sum_{j=1}^{(k-1)(n-1)} \bigg\vert (p_j(\pi) - p_{j+n}(\pi)) \prod_{i \in [n-1]} p_{j+i}(\pi) \bigg\vert \enspace .
\end{align*}

We now prove two supporting propositions that allow us to combine different inputs in the chain.

\begin{proposition} \label{prop:Chaining}
    For all transcripts $\pi$ and $i \in [2, n], j \in [k-1]$ resulting in valid indices,
    \[
        (p_i(\pi) - p_{jn+i}) \prod_{\ell \in [n] \setminus \{i\}} p_\ell(\pi) \quad \leq \quad mn s(\pi)
    \]
\end{proposition}
\begin{proof}
Fix $i \in [n]$ and a transcript $\pi$, and denote $p_j \coloneqq p_j(\pi), s \coloneqq s(\pi)$. Since all $p_j \geq 0$, we assume $\prod_{j \in [n]} p_j \geq ms$, or the proposition follows trivially. For all $j$ resulting in valid indices, define
\[
    q_j \quad \coloneqq \quad \begin{cases}
        p_j & j \leq n \\
        q_{j-n} - s/(\prod_{\ell \in [n-1]} q_{j-\ell}) & j > n
    \end{cases} \enspace .
\]

We first prove some properties about $q_j$.

\begin{lemma} \label{lemma:NTupleRelation}
    For all $j$ resulting in valid indices,
    \[
        \prod_{\ell \in [n]} q_{j+\ell+1} \quad = \quad \prod_{\ell \in [n]} q_{j+\ell} - s \enspace .
    \]
\end{lemma}
\begin{proof}
Follows by definition.
\end{proof}

\begin{lemma} \label{lemma:NTupleBound}
    For all $j$ resulting in valid indices,
    \[
        \prod_{\ell \in [n]} q_{j+\ell} \quad \geq \quad (m-j)s \enspace .
    \]
\end{lemma}
\begin{proof}
Follows from the assumption $\prod_{\ell \in [n]} q_\ell \geq ms$ and repeated applications of \Cref{lemma:NTupleRelation}.
\end{proof}

\begin{lemma} \label{lemma:SingleBounds}
    For all $j$ resulting in valid indices, $q_j \in (0, p_j]$.
\end{lemma}
\begin{proof}
Since $\prod_{\ell \in [n]} q_\ell \geq ms$, $q_j > 0$ for all $j \in [n]$. Then because $\prod_{\ell \in [n]} q_{j+\ell} > 0$ by \Cref{lemma:NTupleBound}, we can iteratively conclude that each subsequent $q_j > 0$.

Now suppose for induction that $q_j \leq p_j$ for all $j \leq j'$. Then
\[
    q_{j'} \quad = \quad q_{j'-n} - \frac{s}{\prod_{\ell \in [n-1]} q_{j'-\ell}} \quad \leq \quad p_{j'-n} - \frac{s}{\prod_{\ell \in [n-1]} p_{j'-\ell}} \quad \leq \quad p_{j'} \enspace ,
\]
where the first inequality follows by induction, and the second inequality follows because
\[
    (p_{j'-n} - p_{j'}) \prod_{\ell \in [n-1]} p_{j'-\ell} \quad \leq \quad \bigg\vert (p_{j'-n} - p_{j'}) \prod_{\ell \in [n-1]} p_{j'-\ell} \bigg\vert \quad \leq \quad s \enspace . \qedhere
\]
\end{proof}

\begin{lemma} \label{lemma:SingleDecreasing}
    For all $i \in [n]$, $q_{jn+i}$ is decreasing in $j$.
\end{lemma}
\begin{proof}
Follows by definition because all $q_j > 0$ (\Cref{lemma:SingleBounds}).
\end{proof}

We now prove the main technical lemma involving the $q_j$.

\begin{lemma} \label{lemma:WhatThe}
    For all $j$ resulting in valid indices,
    \[
        (q_i - q_{jn+i}) \prod_{\ell \in [n] \setminus \{i\}} q_\ell(\pi) \quad \leq \quad jns \enspace .
    \]
\end{lemma}
\begin{proof}
Assume for induction that the claim holds for $j \leq j'$ (the base case $j = 0$ is trivial). Then
\begin{align*}
    (q_i - q_{(j'+1)n+i}) \prod_{\ell \in [n] \setminus \{i\}} q_\ell \quad &= \quad \bigg(q_i - q_{j'n+i} + \frac{s}{\prod_{\ell \in [n-1]} q_{j'n+i+\ell}}\bigg) \prod_{\ell \in [n] \setminus \{i\}} q_\ell(\pi) \\
    &= \quad j'ns + \frac{\prod_{\ell \in [n] \setminus \{i\}} q_\ell}{\prod_{\ell \in [n-1]} q_{j'n+i+\ell}} s \\
    &= \quad j'ns + \frac{q_{j'n+i} \prod_{\ell \in [n] \setminus \{i\}} q_\ell}{\prod_{\ell \in [n]} q_{j'n+i-1+\ell}} s \\
    &\leq \quad j'ns + \frac{\prod_{\ell \in [n]} q_\ell}{\prod_{\ell \in [n]} q_{j'n+i-1+\ell}} s \\
    &= \quad j'ns + \frac{\prod_{\ell \in [n]} q_\ell}{\prod_{\ell \in [n]} q_\ell - (j'n+i-1)s} s \\
    &\leq \quad j'ns + \frac{\prod_{\ell \in [n]} q_\ell}{\prod_{\ell \in [n]} q_\ell - k(n-1)s} s \\
    &\leq \quad j'ns + \frac{kns}{kns - k(n-1)s} s \\
    &= \quad (j'+1)ns \enspace ,
\end{align*}
where the second line follows by the induction hypothesis, the fourth line follows because $q_i \geq q_{j'n+i}$ (\Cref{lemma:SingleDecreasing}), the fifth line follows by iterative applications of \Cref{lemma:NTupleRelation}, the sixth line follows because $(j'+1)n+i$ must be a valid index, and the seventh line follows because $x/(x-y)$ is decreasing in $x$ for all $x > y$, and $\prod_{\ell \in [n]} q_\ell \geq ms = kns$.
\end{proof}

Since $q_j = p_j$ for $j \in [n]$, and $q_j \leq p_j$ for all $j$ (\Cref{lemma:SingleBounds}), we conclude by \Cref{lemma:WhatThe} that
\[
    (p_i - p_{jn+i}) \prod_{\ell \in [n] \setminus \{\ell\}} p_\ell \quad = \quad (q_i - p_{jn+i}) \prod_{\ell \in [n] \setminus \{\ell\}} q_\ell \quad \leq \quad (q_i - q_{jn+i}) \prod_{\ell \in [n] \setminus \{\ell\}} p_\ell \quad \leq \quad mns \enspace . \qedhere
\]
\end{proof}

\begin{proposition} \label{prop:Rectangling}
    Let $\vb{X}, \vb{Y}, \vb{Z} \in (2^M)^n$ and let $I, J \subseteq [n]$ be disjoint sets such that $\vb{X}$ differs from $\vb{Y}$ in coordinates $I$, and $\vb{X}$ differs from $\vb{Z}$ in coordinates $J$. Let $\vb{W}$ be defined $W_i = Y_i$ when $i \in I$, $W_i = Z_i$ when $i \in J$, and $W_i = X_i$ otherwise. Then
    \[
        \TV(\Pi(\vb{X}), \Pi(\vb{W})) \quad \leq \quad \TV(\Pi(\vb{X}), \Pi(\vb{Y})) + \TV(\Pi(\vb{X}), \Pi(\vb{Z})) \enspace .
    \]
\end{proposition}
\begin{proof}
We start with a fact about numbers.

\begin{lemma} \label{lemma:Numbers}
    For any $x_i, x_j, w_i, w_j \geq 0$, $x_i x_j - w_i w_j \leq (x_i x_j - x_i w_j) + (x_i x_j - w_i x_j)$.
\end{lemma}
\begin{proof}
There are a few cases:
\begin{itemize}[topsep=4pt]
    \item If $x_i < w_i$ and $x_j \geq w_j$, then $x_i x_j - w_i w_j < x_i x_j - x_i w_j$ and $x_i x_j - w_i x_j \geq 0$.

    \item If $x_i \geq w_i$ and $x_j < w_j$, the claim follows by similar reasoning.

    \item If $x_i \geq w_i$ and $x_j \geq w_j$, or $x_i < w_i$ and $x_j < w_j$, then
    \begin{align*}
        0 \quad &\leq \quad (x_i - w_i)(x_j - w_j) \\
        0 \quad &\leq \quad x_i x_j - w_i x_j - x_i w_j + w_i w_j \\
        x_i x_j - w_i w_j \quad &\leq \quad x_i x_j - w_i x_j + x_i x_j - x_i w_j \enspace . \qedhere
    \end{align*}
\end{itemize}
\end{proof}

Let $\calT$ be the collection of transcripts $\pi$ such that $\Pr[\Pi(\vb{X}) = \pi] \geq \Pr[\Pi(\vb{W}) = \pi]$. By the definition of total variation distance and \Cref{lemma:RandomizedRectangle,lemma:Numbers},
\begin{align*}
    & \TV(\Pi(\vb{X}), \Pi(\vb{W})) \\
    &\qquad\quad = \quad \sum_{\pi \in \calT} \Big(\Pr[\Pi(\vb{X}) = \pi] - \Pr[\Pi(\vb{W}) = \pi]\Big) \\
    &\qquad\quad = \quad \sum_{\pi \in \calT} \bigg(\prod_{i \in [n]} \rho_i(X_i, \pi) - \prod_{i \in [n]} \rho_i(W_i, \pi)\bigg) \\
    & \qquad\quad = \quad \sum_{\pi \in \calT} \bigg(\prod_{i \in I} \rho_i(X_i, \pi) \prod_{i \in J} \rho_i(X_i, \pi) - \prod_{i \in I} \rho_i(W_i, \pi) \prod_{i \in J} \rho_i(W_i, \pi)\bigg) \prod_{i \not\in I \cup J} \rho_i(X_i, \pi) \\
    & \qquad\quad \leq \quad \sum_{\pi \in \calT} \bigg(\bigg(\prod_{i \in I} \rho_i(X_i, \pi) \prod_{i \in J} \rho_i(X_i, \pi) - \prod_{i \in I} \rho_i(W_i, \pi) \prod_{i \in J} \rho_i(X_i, \pi)\bigg)\prod_{i \not\in I \cup J} \rho_i(X_i, \pi) \\
    & \qquad\qquad\qquad\qquad\quad + \bigg(\prod_{i \in I} \rho_i(X_i, \pi) \prod_{i \in J} \rho_i(X_i, \pi) - \prod_{i \in I} \rho_i(X_i, \pi) \prod_{i \in J} \rho_i(W_i, \pi)\bigg)\prod_{i \not\in I \cup J} \rho_i(X_i, \pi)\bigg) \\
    & \qquad\quad \leq \quad \TV(\Pi(\vb{X}), \Pi(\vb{Y})) + \TV(\Pi(\vb{X}), \Pi(\vb{Z})) \enspace . \qedhere
\end{align*}
\end{proof}

We are now ready to prove the theorem. For all $i \in [n]$, let $\vb{X}_i$ be $\vb{X}^{(0)}$ with the $i$\textsuperscript{th} coordinate replaced by $X^{(1)}_i$ (note that $\vb{X}^{(0)} = \vb{X}_1$), and let $\calT_i$ be the set of transcripts $\pi$ such that $\Pr[\Pi(\vb{X}^{(0)}) = \pi] \geq \Pr[\Pi(\vb{X}_i) = \pi]$. Then for all $i \in [2, n]$,
\begin{align*}
    \TV(\Pi(\vb{X}^{(0)}), \Pi(\vb{X}_i)) \quad &= \quad \sum_{\pi \in \calT_i} \Big(\Pr[\Pi(\vec{X}_1) = \pi] - \Pr[\Pi(\vec{X}_i) = \pi]\Big) \\
    &= \quad \sum_{\pi \in \calT_i} \bigg((p_i(\pi) - p_{(k-1)(i-1)+i}(\pi)) \prod_{\ell \in [n] \setminus \{i\}} p_\ell(\pi)\bigg) \\
    &\leq \quad \sum_{\pi \in \calT_i} mns(\pi) \\
    &\leq \quad mn\sum_\pi s(\pi) \\
    &= \quad mn\sum_{j=1}^{(k-1)(n-1)} \sum_\pi \bigg\vert (p_j(\pi) - p_{j+n}(\pi)) \prod_{\ell \in [n-1]} p_{j+\ell} \bigg\vert \\
    &= \quad 2mn\sum_{j=1}^{(k-1)(n-1)} \TV(\Pi(\vb{L}_{n+j}), \Pi(\vb{L}_{n+j+1})) \\
    &< \quad 2mn \cdot m \cdot \sqrt{\frac{2}{m^8}} \\
    &< \quad \frac{3n}{m^2} \enspace ,
\end{align*}
where the third line is by \Cref{prop:Chaining} (because $k - 1$ is a multiple of $n$), and the second to last line is by \Cref{lemma:UnbrokenChainMeansCloseLinks}.

Finally, let $\vb{X}_{\to i}$ denote $\vb{X}^{(0)}$ with the first $i$ coordinates replaced by those of $\vb{X}^{(1)}$. For each $i \in [n-1]$ in sequence, we apply \Cref{prop:Rectangling} to conclude that
\begin{align*}
    \TV(\Pi(\vb{X}^{(0)}), \Pi(\vb{X}_{\to (i+1)})) \quad &\leq \quad \TV(\Pi(\vb{X}^{(0)}), \Pi(\vb{X}_{\to i})) + \TV(\Pi(\vb{X}^{(0)}), \Pi(\vb{X}_{i+1})) \\
    &\leq \quad \TV(\Pi(\vb{X}^{(0)}), \Pi(\vb{X}_{\to i})) + \frac{3n}{m^2} \enspace .
\end{align*}

Chaining these inequalities together yields
\[
    \TV(\Pi(\vb{X}^{(0)}), \Pi(\vb{X}^{(1)})) \quad = \quad \TV(\Pi(\vb{X}^{(0)}), \Pi(\vb{X}_{\to n})) \quad \leq \quad \frac{3n^2}{m^2} \quad \leq \quad \frac{3}{m} \enspace . \qedhere
\]
\end{proof}

\begin{corollary} \label{cor:LowInformationMeansBadProtocol}
    If $\CIC_\mu(\Pi) \leq 1/m^{n^2+9}$, then $\Pi$ has worst-case error $1/2-O(1/m)$. In other words, if $\Pi$ has worst-case error bounded away from $1/2$, then $\CIC_\mu(\Pi) > 1/m^{n^2+9}$.
\end{corollary}
\begin{proof}
Let $\calT$ be the set of transcripts for which $\Pi$ outputs $1$.

Observe that $X^{(1)}_i = L_{(k-1)(i-1)+n+1,i} = S_{(k-1)(i-1)+i} = S_{k(i-1)+1}$. In other words, $\vb{X}^{(1)} = (S_1, S_{1+k}, S_{1+2k}, \dots, S_{1+(n-1)k})$ is comprised of disjoint sets, so $\FS{m,n}(\vb{X}^{(1)}) = 1$. On the other hand, $\vb{X}^{(0)}$ is a link, so $\FS{m,n}(\vb{X}^{(0)}) = 0$.

Therefore, we have two instances for which the total variation distance between the transcripts is $O(1/m)$, so $\abs{\Pr[\Pi(\vb{X}^{(0)}) \in \calT] - \Pr[\Pi(\vb{X}^{(1)}) \in \calT]} = O(1/m)$. Thus, $\Pi$ has $O(1/m)$ advantage distinguishing between instances $\vb{X}^{(0)}, \vb{X}^{(1)}$ and hence worst-case error at least $1/2-O(1/m)$.
\end{proof}

\subsection{Wrapping Up}

Before we can conclude, we need to ensure that $\mu$ satisfies the conditions of \Cref{thm:DirectSum}. Conditions $(2)$ and $(3)$ are satisfied by construction, so we just need to verify condition $(1)$. Let $\calX$ be the promise of \FS{m,n} and $\calX^*$ be the promise of \EFS{m,n,z}.

\begin{proposition} \label{prop:Safety}
    For $z = 2^{\Theta(\sqrt{m})}$ and $\varepsilon = 2^{-\Theta(\sqrt{m})}$, $\mu, \calX$ are $(z, \varepsilon)$-safe with respect to $\calX^*$.
\end{proposition}
\begin{proof}
There are two promises of \EFS{} not present in \FS{}: 
\begin{itemize}[topsep=4pt]
    \item Promise 3: $\vb{\closure{X}}_i$ is $\lambda \coloneqq \log(m^{1/3})/\log(2n)$-sparse.

    \item Promise 4: For all $i \in [n]$ and $S \subseteq M$ of size $n\log(m)$, there exists $\ell \in [z]$ such that $S \subseteq \closure{X}_{i,\ell}$.
\end{itemize}

Fix any player, and let $\vb{X} \in (2^M)^z$ be their input. For some $\ell' \in [z]$, $X_{\ell'}$ is an arbitrarily chosen set of size $k$, and all other $X_{\ell}$ are independent uniformly random sets of size $k$. Observe that both promises are preserved under permuting the items, so by applying a uniformly random permutation to the items, $X_{\ell'}$ also becomes an independent uniformly random set of size $k$. Therefore, we may WLOG assume that $\vb{X}$ is a list of independent, uniformly random sets of size $k$.

\begin{lemma} \label{lemma:Sparsity}
    $\vb{\closure{X}}$ is not $\lambda$-sparse w.p. at most $2^{-\Omega(\sqrt{m})}$.
\end{lemma}
\begin{proof}
Suppose $n \leq m^{1/3}$; otherwise, $\lambda = 1$ and the statement holds vacuously.

We sample $X_1, \dots, X_z$ in a roundabout way: let $Y_1, \dots, Y_z \subseteq M$ be sets sampled by independently including each item w.p. $1/(2n)$. If any $Y_\ell$ are of size more than $k$, we abort. Otherwise, we let each $X_\ell$ be a uniformly random superset of $Y_\ell$. Observe that since $\closure{Y}_\ell \supseteq \closure{X}_\ell$ for all $\ell \in [z]$, $\lambda$-sparsity of $\vb{\closure{Y}}$ implies $\lambda$-sparsity of $\vb{\closure{X}}$.

Let $\calL \subseteq 2^{[z]}$ be the collection of all sets of size $\lambda$. Then using independence and union bounds, the probability that $\vb{\closure{Y}}$ is \emph{not} $\ell$-sparse is
\begin{align*}
    \Pr\bigg[\bigcup_{L \in \calL} \bigg\{\bigcup_{\ell \in L} \closure{Y}_\ell = M\bigg\}\bigg] \quad &= \quad \Pr\bigg[\bigcup_{L \in \calL} \bigcap_{j \in M} \bigcup_{\ell \in L} \{j \in \closure{Y}_\ell\}\bigg] \\
    &= \quad \Pr\bigg[\bigcup_{L \in \calL} \bigcap_{j \in M} \closure{\bigcap_{\ell \in L} \{j \in Y_\ell\}}\bigg] \\
    &\leq \quad \abs{\calL} \bigg(1 - \bigg(\frac{1}{2n}\bigg)^\lambda\bigg)^m \\
    &= \quad \abs{\calL} \bigg(1 - \frac{1}{m^{1/3}}\bigg)^m \\
    &\leq \quad z^\lambda e^{-m^{2/3}} \quad = \quad 2^{-\Omega(m^{2/3})} \enspace .
\end{align*}

Additionally, by a Chernoff bound and union bound, the probability that some $Y_\ell$ is of size more than $k$ is at most $z e^{-\Theta(k)} \leq 2^{-\Omega(\sqrt{m})}$. The lemma follows by a union bound over the two events.
\end{proof}

\begin{lemma} \label{lemma:SmallCoverage}
    There exists $S \subseteq M$ of size $n\log(m)$ such that $S \not\subseteq \closure{X}_\ell$ for all $\ell \in [z]$ w.p. at most $2^{-\Omega(\sqrt{m})}$
\end{lemma}
\begin{proof}
We sample $X_1, \dots, X_z$ in a slightly different roundabout way: let $Y_1, \dots, Y_z \subseteq M$ be sets sampled by independently including each item w.p. $3/(2n)$. If any $Y_\ell$ are of size less than $k$, we abort. Otherwise, we let each $X_\ell$ be a uniformly random subset of $Y_\ell$. Observe that since $\closure{Y}_\ell \subseteq \closure{X}_\ell$ for all $\ell \in [z]$, the lemma follows for $\vb{X}$ so long as it follows for $\vb{Y}$.

Let $\calS \subseteq 2^M$ be the collection of all sets of size $n\log(m)$. Then using independence and union bounds, the probability that there exists $S \in \calS$ which is not contained in $\closure{Y}_\ell$ for any $\ell \in [z]$ is
\begin{align*}
    \Pr\bigg[\bigcup_{S \in \calS} \bigcap_{\ell \in [z]} \{S \not\subseteq \closure{Y}_\ell\}\bigg] \quad &= \quad \Pr\bigg[\bigcup_{S \in \calS} \bigcap_{\ell \in [z]} \{S \cap Y_\ell \ne \emptyset\}\bigg] \\
    &= \quad \Pr\bigg[\bigcup_{S \in \calS} \bigcap_{\ell \in [z]} \bigcup_{j \in S} \{j \in Y_\ell\}\bigg] \\
    &= \quad \Pr\bigg[\bigcup_{S \in \calS} \bigcap_{\ell \in [z]} \closure{\bigcap_{j \in S} \{j \not\in Y_\ell\}}\bigg] \\
    &\leq \quad \abs{\calS} \bigg(1 - \bigg(1 - \frac{3}{2n}\bigg)^{n\log(m)}\bigg)^z \\
    &\leq \quad \abs{\calS} (1 - e^{-3\log(m)})^z \\
    &\leq \quad m^{n\log(m)} e^{z/m^3} \quad \ll \quad 2^{-\Omega(\sqrt{m})} \enspace .
\end{align*}

Additionally, by a Chernoff bound and union bound, the probability that some $Y_\ell$ is of size less than $k$ is at most $z e^{-\Theta(k)} \leq 2^{-\Omega(\sqrt{m})}$. The lemma follows by a union bound over the two events.
\end{proof}

The result follows by a union bound over the bad events in \Cref{lemma:Sparsity,lemma:SmallCoverage} and a union bound over the $n$ players.
\end{proof}

\begin{proof}[Proof of \Cref{thm:EFS}]
For $z = 2^{O(\sqrt{m})}$, we have by \Cref{prop:Safety} and \Cref{cor:DirectSum,cor:LowInformationMeansBadProtocol} that the randomized communication complexity of solving \EFS{m,n,z} with worst-case error $1/2-O(1/m)-2^{-\Omega(\sqrt{m})}$ is $z \cdot \CIC_\mu(\FS{m,n}) = 2^{\Omega(\sqrt{m}-n^2\log(m))}$.
\end{proof}

\bibliographystyle{alpha}
\bibliography{MasterBib}

\appendix

\section{Missing Proofs}
\label{sec:missing-proofs}

\TwoProblems*

\begin{proof}
Suppose $\calQ$ can be solved w.p. at least $2/3$ in $z'$ communication. Taking parallel repetitions, $\calQ$ can be solved w.p., at least $0.99$ in $O(z')$ communication. Denote any such protocol by $\Pi$, and let $\Pi(\vb{x}, \vb{y}) \in \{0, 1\}$ be the output of $\Pi$ on input $(\vb{x}, \vb{y})$.

Let $\mu$ be any distribution over $\calY^*$. Define $\mu_0$ to be $\mu$ conditioned on sampling a $0$-instance, and $\mu_1$ to be $\mu$ conditioned on sampling a $1$-instance. Let $\Pi_\mu$ to be the following protocol for $\calQ_f$:
\begin{enumerate}[topsep=4pt]
    \item Let $\vb{x}$ be the input.

    \item Sample $\vb{Y}^{(0)}_1, \dots, \vb{Y}^{(0)}_{10} \sim \mu_0$ and $\vb{Y}^{(1)}_1, \dots, \vb{Y}^{(1)}_{10} \sim \mu_1$.

    \item Run $\Pi$ on inputs $\{(\vb{x}, \vb{Y}^{(b)}_\ell) : b \in \{0, 1\}, \ell \in [10]\}$.

    \item If $P_0 \coloneqq \Pi(\vb{x}, \vb{Y}^{(0)}_1) + \dots + \Pi(\vb{x}, \vb{Y}^{(0)}_{10}) \ne 0$, output $1$.

    \item Otherwise, if $P_1 \coloneqq \Pi(\vb{x}, \vb{Y}^{(1)}_1) + \dots + \Pi(\vb{x}, \vb{Y}^{(1)}_{10}) \ne 10$, output $0$.

    \item Otherwise, abort.
\end{enumerate}

Since $\Pi_\mu$ just runs $\Pi$ on $20$ inputs, $\Pi_\mu$ runs in $O(z')$ communication.

If there exists $\mu$ such that $\Pi_\mu(\vb{x}) = f(\vb{x})$ w.p. at least $2/3$ for all $\vb{x} \in \calX^*$, then $\Pi_\mu$ solves $\calQ_f$ w.p. at least $2/3$, and hence $\Pi_\mu$ must use $\Omega(z)$ communication.

Otherwise, for every $\mu$, there exists $\vb{x}_\mu \in \calX^*$ such that $\Pi_\mu(\vb{x}_\mu) \ne f(\vb{x}_\mu)$ w.p. at least $1/3$.

\begin{lemma} \label{lemma:DistributionalComplexity}
    If $\Pi_\mu(\vb{x}_\mu) \ne f(\vb{x}_\mu)$ w.p. at least $1/3$, then $\Pr_{\vb{Y} \sim \mu}[\Pi(\vb{x}_\mu, \vb{Y}) = g(\vb{Y})] \geq 2/3$.
\end{lemma}
\begin{proof}
Suppose that $\Pr_{\vb{Y} \sim \mu}[\Pi(\vb{x}_\mu, \vb{Y}) = g(\vb{Y})] < 2/3$ for contradiction.

If $f(\vb{x}_\mu) = 1$, then $\Pr_{\vb{Y} \sim \mu_1}[\Pi(\vb{x}_\mu, \vb{Y}) = 1] \geq 0.99$ by correctness of $\Pi$, so it must be the case that $\Pr_{\vb{Y} \sim \mu_0}[\Pi(\vb{x}_\mu, \vb{Y}) = 0] < 2/3$.

The protocol can only fail on $\vb{x}_\mu$ if it outputs $0$ at Step $(5)$ or aborts at Step $(6)$. Therefore, the protocol only fails if $P_0 = 0$. However, $\Pr[\Pi_\mu(\vb{x}_\mu, \vb{Y}^{(0)}_\ell) = 0] < 2/3$, so $\Pr[P_0 = 0] \leq (2/3)^{10} < 1/3$, a contradiction (as the protocol only fails when $P_0 = 0$).

If $f(\vb{x}_\mu) = 0$, then $\Pr_{\vb{Y} \sim \mu_0}[\Pi(\vb{x}_\mu, \vb{Y}) = 0] \geq 0.99$ by correctness of $\Pi$, so it must be the case that $\Pr_{\vb{Y} \sim \mu_1}[\Pi(\vb{x}_\mu, \vb{Y}) = 1] < 2/3$.

The protocol can only fail if it outputs $1$ at Step $(4)$ or aborts at Step $(6)$. Therefore, the protocol only fails if $P_0 \ne 0$ or $P_1 = 10$. However, $\Pr[\Pi_\mu(\vb{x}_\mu, \vb{Y}^{(1)}_\ell) = 1] < 2/3$, so $\Pr[P_1 = 10] \leq (2/3)^{10} < 0.1$. Further, $\Pr[P_0 \ne 0] \leq 10 \cdot 0.01 = 0.1$. So the protocol fails w.p. at most $0.2$, a contradiction.
\end{proof}

By~\Cref{lemma:DistributionalComplexity}, if there does not exist a single $\mu$ such that $\Pi_\mu(\vb{x})$ solves $\calQ_f$ w.p.~at least $2/3$, then for every distribution $\mu$ over $\calY^*$, the randomized protocol $\Pi(\vb{x}_\mu, \cdot)$ computes $g$ w.p. at least $2/3$ for $\vb{Y} \sim \mu$, using $O(z')$ communication. By Yao's Minimax Principle, this is equal to the communication complexity of solving $\calQ_g$ w.p. at least $2/3$, and hence $z' = \Omega(z)$.
\end{proof}

\section{Code for Separation Results}
\label{sec:code}

\begin{lstlisting}
package separations;

import java.util.Arrays;

public class XOSSM {

	private static final double EPSILON = 0.001;
	private static final double DELTA = 0.001;
	private static final int N = 150;

	public static void main(String[] args) {
		for (int n = 2; n <= N; n++) {
			double approxXOS = approxXOS(n);
			double approxXOSSM = approxXOSSM(n, 1 - 1.0 / (n + 1), DELTA);

			System.out.printf("%.5f\n%.5f\n\n", approxXOS, approxXOSSM);

			if (approxXOSSM < approxXOS + EPSILON) {
				System.out.println("BAD");
				System.exit(0);
			}
		}

		double approxXOS = 1 / (1 - 1 / Math.E);
		double approxXOSSM = approxXOSSM(N, 1 - 1.0 / (N + 1), DELTA);

		System.out.printf("%.5f\n%.5f\n\n", approxXOS, approxXOSSM);

		if (approxXOSSM < approxXOS + EPSILON) {
			System.out.println("BAD");
			System.exit(0);
		}

		System.out.println("Done");
	}

	private static double approxXOS(int n) {
		return 1 / (1 - Math.pow(1 - 1.0 / n, n));
	}

	private static double approxXOSSM(int n, double p, double delta) {
		// max_{t* in [n]} approxXOSSM(t*)
		double[] approxes = new double[n];

		for (int i = 0; i < n; i++) {
			approxes[i] = approxXOSSM(n, i + 1, p, delta);
		}

		return Arrays.stream(approxes).max().getAsDouble();
	}

	private static double approxXOSSM(int n, int tStar, double p, double delta) {
		return instance1(n, tStar, p) / instance0(n, tStar, p) - delta;
	}

	private static double instance1(int n, int tStar, double p) {
		return p * n + (1 - p) * opt(n, tStar);
	}

	private static double instance0(int n, int tStar, double p) {
		// sum_{t=0}^n binom(n, t) (1 - p)^{n - t} p^t max{OPT(t), OPT(t*))}
		double welfare = 0;
		for (int t = 0; t <= n; t++) {
			welfare += binom(n, t) * Math.pow(1 - p, n - t) * Math.pow(p, t) * Math.max(opt(n, t), opt(n, tStar));
		}

		return welfare;
	}

	private static double opt(int n, int t) {
		// OPT(t) = t / alpha(t) = n * (1 - (1 - 1/n)^t)
		return n * (1 - Math.pow(1 - 1.0 / n, t));
	}

	private static int binom(int n, int k) {
		if (k > n - k) {
			k = n - k;
		}

		int b = 1;
		for (int i = 1, m = n; i <= k; i++, m--) {
			b = b * m / i;
		}

		return b;
	}

}
\end{lstlisting}

\end{document}